\documentclass[11pt]{article}
\usepackage[margin = 1.2 in]{geometry}
\usepackage[utf8]{inputenc}
\usepackage{amsfonts}
\usepackage{amsmath}
\usepackage{amssymb}
\usepackage[english]{babel}
\usepackage{amsthm}
\usepackage{tikz}
\usepackage{float}
\usepackage[export]{adjustbox}
\usepackage{subcaption}
\usepackage{comment}
\usepackage{bbm}
\usepackage{natbib}
\usepackage{enumerate}
\usepackage{hyperref}
\usepackage{cleveref}
\setcitestyle{authoryear,open={(},close={)}}
\PassOptionsToPackage{hyphens}{url}\usepackage{hyperref}

\excludecomment{note}

\newtheorem{proposition}{Proposition}
\newtheorem{lemma}{Lemma}

\newtheorem{corollary}{Corollary}
\newtheorem*{definition}{Definition}

\theoremstyle{remark}

\DeclareMathOperator*{\argmax}{argmax}

\newcommand{\vst}{\vspace{3mm}}

\title{Screening and Information-Sharing Externalities\thanks{I am grateful to Marciano Siniscalchi, Eddie Dekel, Piotr Dworczak, Wojciech Olszewski, Alessandro Pavan, Harry Pei, Bruno Strulovici, and Asher Wolinsky for helpful comments and discussion}}
\author{Quitz\'{e} Valenzuela-Stookey\thanks{Department of Economics, Northwestern University}}
\date{November 8, 2020}

\begin{document}
\maketitle

\begin{note} 
    \Large \textcolor{blue}{WITH NOTES}
\end{note}

\begin{center}
    \Large \textcolor{blue}{\href{https://northwestern.box.com/s/lovaravqcvlv29hoxoljrjb6j38kh3lt}{Click here for the latest version}}
\end{center}
\vst

\begin{abstract}
In many settings, multiple uninformed agents bargain simultaneously with a single informed agent in each of multiple periods. For example, workers and firms negotiate each year over salaries, and the firm has private information about the value of workers' output. I study the effects of transparency in these settings; uninformed agents may observe others' past bargaining outcomes, e.g. wages. I show that in equilibrium, each uninformed agent will choose in each period whether to try to separate the informed agent's types (screen) or receive the same outcome regardless of type (pool). In other words, the agents engage in a form of experimentation via their bargaining strategies. There are two main theoretical insights. First, there is a \textit{complementary screening} effect: the more agents screen in equilibrium, the lower the information rents that each will have to pay. Second, the payoff of the informed agent will have a certain supermodularity property, which implies that equilibria with screening are ``fragile'' to deviations by uninformed agents. I apply the results to study pay-secrecy regulations and anti-discrimination policy. I show that, surprisingly, penalties for pay discrimination have no impact on bargaining outcomes. I discuss how this result depends on the legal framework for discrimination cases, and suggest changes to enhance the efficacy of anti-discrimination regulations. In particular, anti-discrimination law should preclude the so-called ``salary negotiation defense''.
\end{abstract}

Many ongoing policy debates relate to transparency in negotiations. These include pay transparency among employees and the disclosure of payer-negotiated rates by hospitals. These settings generally feature asymmetric information. Insurers that negotiate rates with a hospital have only partial knowledge of the hospital's costs for performing different services. Workers are generally unsure about how much revenue they will generate for a potential employer. When there is one-sided asymmetric information and the informed party, say a firm, is long-lived, the firm will have incentives to establish a reputation which will keep its future costs down. For example, a firm may be reluctant to increase wages for current workers, since future workers will infer that their value to the firm is high, and demand higher wages as well. On the other hand, when the uninformed party, say a worker, is also long-lived, they may engage in strategic experimentation to try to learn about the firm's private information. For example, early in their career the worker may take a hard-line position in wage negotiations, in the hopes of learning how valuable they are to the firm. This information will be useful to the worker in future negotiations. When there are multiple workers, each one may learn something about their own value to the firm by observing the wages of others. This gives rise to information externalities. 

The goal of this paper is to understand the interaction between reputation effects and information externalities. In particular, I am interested in the implications of this interaction for policies that affect the transparency of negotiations, i.e. the amount of information shared between the (initially) uninformed parties. I consider a model in which multiple uninformed parties screen a single agent with an unknown type. I will call the uninformed parties ``workers'' and the agent the ``firm''.\footnote{The use of the terms ``worker'' and ``firm'' is for exposition only. The model applies equally well to any setting in which multiple principals simultaneously screen an informed agent, such as hospital-insurer negotiations.} The firm's type determines the output generated by each worker. 

Before describing the model that will be studied in this paper, it will be helpful to establish some basic intuition regarding reputation effects. Reputation effects alone can be understood in a model in which workers arrive sequentially and bargain with the firm over their wage. Consider such a model, in which before negotiating a worker may observe the wages paid to some previous workers. Reputational bargaining games of this sort are studied by \cite{kreps1982reputation}, \cite{fudenberg1989reputation}, and \cite{schmidt1993commitment}. In general, workers who observe that previous workers were paid a high wage will infer that the firm is the high-value type, and demand high wages themselves. Thus the firm will be unwilling to pay high wages, for fear of damaging its reputation. Increases in negotiation transparency, for example example an increase in the probability that each worker observes the wages of those that came before, strengthen the firms reputational incentives, and lead to lower wages.\footnote{Greater information-sharing of this form is analogous to an increase in the discount factor of the firm.} In other words, the information rents that the workers must pay in order to separate the firm types increase in the degree of information sharing. When these rents become too large, workers will give up trying to screen the firm, and settle for a fixed wage independent of type. 

While the effect of observing past bargaining outcomes in the sequential setting is well understood, less is known about the information externalities that arise in situations in which negotiations take place with multiple uninformed parties each period. This is the case in many prominent settings. In the worker-firm example, the firm recruits a new class of junior employees, and negotiate wages with each simultaneously. Similarly, hospitals negotiate with multiple insurers over the the rates to be paid by the insurers for different services, where the hospital has private information about its costs. Information sharing between workers features prominently in ongoing policy debates in these settings. For workers and firms, the costs and benefits of ``pay secrecy'' policies, in which firms prevent employees from sharing or discussing their compensation, are the subject of lively debate.\footnote{\cite{bradford2018}.} Regulation in place since the National Labor Relations Act of 1935 aims to prevent employers from imposing pay-secrecy policies. Transparency of the payer-negotiated rates payed by insurers to hospitals is the subject of new regulations from the Centers for Medicaid and Medicare Services. The new regulation would require hospitals to publicly disclose negotiated rates, and facilitate access to this information by the general public. These requirements are being vigorously resisted by both insurers and hospitals.\footnote{\cite{wilsonpecci2020}.} 

I study a simple model in which multiple workers negotiate with the firm in each period. Both the firm and workers live for two periods. Before the start of the second period, workers may with some exogenous probability observe the wages paid to other workers in the first period. I show that in equilibrium, each worker will adopt one of two types of strategies; a screening strategy, in which the worker structures their offers so as to learn the firm's type in the first period, or a pooling strategy, in which they are paid the same first-period wage regardless of the firm's type. 

The key insight is that in the first period the firm is less averse to being screened by any individual worker when many of the workers are screening, as opposed to only a few screening and the rest pooling. I refer to this as the \textit{complementary screening effect}. To see why this is the case, consider the payoff to the firm from agreeing to a high first-period wage for worker $i$ which reveals the firm's type, i.e. allowing $i$ to screen. Compare the firm's payoff from doing so in two different scenarios: when the only worker, aside from $i$, who screens is $j$, or when workers $j$ and $k$ both screen. There are two factors to consider when comparing these two scenarios: the information flowing to $i$ and the information flowing from $i$. In terms of information flowing from $i$, there are two differences between these scenarios. First, note that information sharing only has an effect on workers that did not screen in the first period (and thus do not know the firm's type when bargaining in the second period). Thus when both $j$ and $k$ are screening, the firm does not need to worry about $k$ observing that $i$ received a high wage, as it would if only $j$ is screening. Second, some workers other than $i$ may observe $k$'s wage. When $k$ is screening these workers learn the firm's type, and so any information revealed to them by observing $i$'s wage is redundant. In other words, the information revealed by $i$'s wage is less likely to be pivotal. There is, moreover, one difference regarding the information flowing to $i$. Should the firm reject $i$'s offer, i.e. prevent $i$ from screening, $i$ is more likely to learn the firm's type anyway when both $j$ and $k$ are screening then when $j$ alone is screening. Thus in the scenario with both $j$ and $k$ screening, the firm's payoff from rejecting $i$'s screening offer is lower. All three of these differences point in the same direction: the firm is more willing to be screened by $i$ when $j$ and $k$ are also screening than when $j$ is screening and $k$ is not. Because of this complementary screening effect, there is an externality from screening that operates through the information rents that must be paid to the high-type firm in order to screen; which are lower the more workers screen. This is in addition to the purely informational externality: workers who do not screen may benefit from the information generated by a screening worker.

Formally, complementary screening manifests as supermodularity in the payoffs of the high-type firm, as a function of the set of workers that it allows to screen in the first period. This has a number of interesting implications. First, screening is fragile; if any worker who is expected to screen in the first period deviates, either by demanding too high of a wage or switching to a pooling strategy, the high type firm will mimic the low type firm with all workers. In other words, screening breaks down completely. Second, screening by some workers will only be possible in equilibrium if enough workers engage in screening. If only a small number of workers attempt to screen, the information rents needed to induce the high-type firm to reveal itself would be too high. 

I also consider a model in which workers differ in how costly it is for them to screen the firm. I show how the distribution of worker types affects wages. I discuss how policies that increase wage transparency may have the adverse effect of destroying screening, and thus reducing the amount of information about the firm type that is generated. These negative effects can be offset however by policies that encourage more workers to screen. Importantly, there are increasing returns to such interventions; the more workers engage in screening, the less costly it is to induce them to screen. 

Finally, I incorporate discrimination into the model of pay secrecy. Reducing pay discrimination is one of the primary motivations for increasing pay transparency. I show how the legal framework within which discrimination cases are tried affects the efficacy of penalties for discrimination. An important factor is whether the law admits the so-called ``salary negotiation defense''. If the law admits such a defense, a worker in a protected group must demonstrate not only that they were paid less than a coworker, but also that the firm rebuffed their attempt to negotiate a higher wage. Surprisingly, penalties for discrimination have no impact on the incentives of discriminatory high-type firms under such laws. This is due to the supermodularity in firm payoffs discussed above: if a firm rejects an equilibrium screening offer, or if it rejects a screening offer from a pooling worker who deviates to screening, then it will reject all screening initial offers. As a result, there will never be verifiable cases of pay discrimination, and hence the penalty in such cases is irrelevant. This implies that pay transparency will have not direct effect on discrimination by high-type firms. On the other hand, high penalties and high transparency help prevent discrimination by low-type firms. I discuss how these findings relate to empirical evidence. The results suggest ways in which discrimination should be defined in order to restore the efficacy of anti-discrimination penalties. 

\subsubsection*{Related literature}

This paper combines elements from the strategic experimentation and reputational bargaining literatures. Conceptually, this paper is closely related to the large literature on strategic experimentation with multiple agents, in which information externalities also arise. Screening by workers in the first period can be thought of as costly experimentation. In their seminal paper, \cite{bolton1999strategic} study a multi-agent version of the multi-armed bandit problem. When there are multiple agents experimenting and observing each other's signals, free riding and encouragement effects arise. \cite{keller2005strategic} study a version of this problem with exponential bandits. Other papers, such as \cite{murto2011learning}, study similar problems in which only actions are observed. One key difference between this literature and the current paper is that in my model the cost of experimentation is endogenously determined by the information rents that must be paid to the firm. This is the source of the complementary screening affect, and its interesting implications for transparency policy. These indirect externalities are distinct from the direct payoff externalities in some experimentation models, such as those of \cite{strulovici2010learning} and \cite{thomas2020strategic}. Moreover, the fact that information is revealed by the strategic firm means that each agents actions can have contemporaneous effects on the information received by other workers. For example, when a screening worker $i$ deviates to pooling, they anticipate that the high type firm will reject the screening offers of all other screening workers (Lemma \ref{lem:simple_screendeviate2}). This affects both the information that $i$ expects to receive from other workers and the payoffs of other workers. 

There is also a close connection between this paper and the reputational bargaining literature (early contributions include \cite{fudenberg1989reputation}, \cite{schmidt1993commitment}). In both cases, a long-lived principal, in this case the firm, has an incentive to take actions which influence the beliefs of short-lived agents (workers) regarding its type. \cite{chaves2019privacy} studies the role of privacy in reputation games. Much of this literature assumes the existence of so called ``commitment types'' of the informed player. More recent papers, such as \cite{pei2020trust}, study reputation without commitment types. In contrast to most of this literature, I study a situation in which \textit{i}) both the informed and uninformed players are long-lived, and \textit{ii}) there are multiple uninformed agents who negotiate simultaneously with the informed player. Combining these elements yields the novel insights on informational externalities and complementary screening. Additionally, there are no direct payoff externalities. \cite{fudenberg1987reputation} and \cite{ghosh2014multiple} share some of these features. However these papers focus on infinite time horizons and conditions under which the informed player can obtain their Stackleberg payoff. 

This paper is also related to the literature on bargaining with incomplete information, including the important contributions of \cite{grossman1986perfect}. For an overview of this literature see \cite{ausubel2002bargaining}. It also contains elements of bilateral bargaining games, such as those studied by \cite{stole1996intra} and \cite{collard2019nash}.

The remainder of this paper is organized as follows. Section \ref{sec:simple} presents the model, and introduces some preliminary results. Section \ref{sec:equilibrium} characterizes equilibrium play. In Section \ref{sec:worker_heterogeneity} I examine some implications of worker heterogeneity. Section \ref{sec:paysecrecy} presents the application to pay secrecy policy with discriminatory firms. Appendix \ref{sec:alternatingoffers} discusses an extension to a model in which bargaining occurs via an alternating offers game \`a la \cite{rubinstein1982perfect}.

\section{A simple bargaining model}\label{sec:simple}

There is a set $W$ of workers, one firm, and 2 periods. The firm has private information about the per-period output generated by a worker, which is $s \in \{s', s''\}$, with $s'' > s' > 0$. The firm's total output is additively separable across workers. The worker's outside option is $-d < 0$, so there are always gains from cooperation. The worker's prior belief is that the firm is of type $s''$ with probability $p$. Workers and firms discount at rate $\beta$ between periods. 

I present here the simplest bargaining game which illustrates the important forces involved in bargaining with multiple workers.\footnote{The general qualitative results regarding the complementarity of screening do not depend on the specific details of the stage game between the workers and the firm. For clarity I will focus on a simple bargaining game. In the appendix I show how similar results obtain when the bargaining game is alternating offers \`{a} la \cite{rubinstein1982perfect}. The essential feature of the stage game is that the worker chooses between a risky but informative ``screening'' action and a safe but uninformative ``pooling'' action.} There are two rounds of bargaining in each period. The worker makes the offer in each round. If the firm rejects both offers then players get their outside option of zero for that period (if this occurs in the first period the parties negotiate again in the second). Both workers and the firm discount at rate $\delta$ between bargaining rounds, so there is a cost to delaying agreement until the second round of offers. To simplify the analysis, I assume that $s' > ps' - d$. This means that a worker who has no information will make an offer of $s'$ in the second round.

Information sharing is as follows. Workers observe nothing of the negotiations between the firm and other workers within each period. However after the first period, but before the second period wage negotiations, worker $i$ observes worker $j$'s first period wage with probability $\rho_{ij}$ (so $i$'s observations are independent across workers). Wage observations are hard evidence; the worker can prove to the firm that they have observed a give wage. The solution concept is PBE. 
 
\subsection{Preliminary observations}

I say that a worker screens in a given period if, in equilibrium, it learns the type of the firm. The worker pools if it does not update it's belief about the firm's type. Partial screening, which will occur when the firm follows a mixed strategy, will be discussed later on.
 
First, observe the high-type firm has more to lose from delaying agreement. Thus whenever a worker screens it must be that only the high-type firm accepts the initial offer.

\begin{lemma}\label{lem:simple_singlecross}
Any worker that screens makes an initial offer $w_1$ that it expects only the type $s''$ firm to accept. 
\end{lemma}
\begin{proof}
The type $s'$ firm will never accept an offer above $s'$. If in equilibrium only the low type firm accepts the initial offer $w_1^i$ made by worker $i$ then worker $i$ will offer $s''$ in the second round. If this is occurring in the first period, worker $i$, as well as any workers observing $i$'s fist period wage, will offer $s''$ in both rounds in the second period. Thus the type $s''$ firm would prefer to accept $i$'s initial offer of $w_1^i < s''$. 
\end{proof}

Consider first the second period equilibrium. From this point on there is no information sharing, so the problem separates completely across workers. In the second round of the second period negotiation, a worker with belief $q$ will either offer a wage of $s''$ (if $qs'' \geq s'$) or $s'$ (if $qs'' < s'$). In either case the low type firm would receive a payoff of 0. Thus the low type firm strictly prefers to accept any first round offer $w_1 < s'$ than to reject such an offer. This has the following implication.

\begin{lemma}\label{lem:simple_pooling2}
In the second period negotiation, the unique pooling equilibrium wage is $s'$. 
\end{lemma}
\begin{proof}
Suppose that a there is a pooling equilibrium in which the worker makes an initial offer of $w_1 < s'$ and both types accept. If the worker deviates and makes an initial offer of $w' \in (w_1,s')$ the low type firm strictly prefers acceptance to rejection. Thus if the firm rejects $w'$ the worker will believe that it is type $s''$, and make a wage offer of $s''$ in the second period. Therefore both firm types accept $w'$, so $w_1$ cannot be an equilibrium offer. Thus the unique pooling equilibrium wage is $s'$
\end{proof}

Suppose the worker screens the firm. Lemma \ref{lem:simple_singlecross} says that the high type firm will be the one to accept first. The type $s''$ firm will accept an initial offer of $w_1$ iff $s'' - w_1 \geq \delta(s'' - s')$. So the equilibrium screening wage is $w^s = \delta s' + (1-\delta)s''$. The worker with belief $q$ prefers screening to pooling in the second period iff $q w^s + (1-q)\delta s' \geq s'$, or equivalently, $qs'' \geq s'$. Notice that the low type firm always gets a payoff of zero in the second period, regardless of whether the worker screens or pools. Let $U_2(s'')$ be the equilibrium second period payoff of a type $s''$ firm when the worker has belief $p$, so $U_2(s'') = \delta(s''-s')$ if the worker screens and $s''-s'$ if the worker pools.

Turn now to the negotiations in the first period. To simplify the analysis, and focus attention on the implications of information sharing, assume that no worker will ever offer a wage below $s'$. I refer to this as the minimum-wage restriction. Offers below $s'$ could arise in equilibrium if no more restrictions are placed on off-path beliefs. This is due to the inferences made by workers who observe another's wage, combined with the fact that the entire negotiation is not observable.\footnote{If there is no information sharing (i.e. $\rho_{i,j} = 0 \forall i,j$), or alternatively if the information shared includes all the details of the negotiation, rather than just the agreed-upon wage, then no wages less that $s'$ will be offered in equilibrium.} 

Under the minimum-wage restriction, it will never be worthwhile for the firm to reject the offer of $s'$ workers who are supposed to pool on-path. 

\begin{lemma}\label{lem:simple_pool}
Under the minimum-wage restriction, if a worker pools on-path then the firm never benefits from rejection of the initial offer $s'$.
\end{lemma}
\begin{proof}
This is obvious for the low-type firm, so consider a high-type firm. By rejecting the initial offer of $s'$ from some pooling worker $i$, the high type firm can only hope to gain by convincing some set of workers not to screen in the second period. If the firm accepts a second-round offer $w > s'$ then all workers who observe this will infer that the firm is type $s'$ for sure, so this cannot be beneficial. If instead the firm accepts a second-round offer of $s'$ then all workers other than $i$ who observe this will make the equilibrium inference and retain their prior. The gains from convincing $i$ not to screen in the second period are sufficient to justify the delay in the first period from the first to the second round iff
\begin{equation*}
    \delta(s'' - s') + \beta(s'' - s') \geq (s'' - s') + \beta \delta(s'' - s'),
\end{equation*}
or equivalently $\delta + \beta \geq  1 + \beta \delta$. This holds with equality if $\beta = 1$, but is violated for any $\beta < 1$ and $\delta < 1$. Thus the firm will never find it beneficial to reject the initial offer. 
\end{proof}

\section{Equilibrium}\label{sec:equilibrium}

I turn now to characterizing equilibrium play in the full game. Consider first the incentives of workers who screen in the first period, and of the high-type firm confronted with screening offers. I will restrict attention to pure strategies. This is not without loss; there exist mixed strategy equilibria, and these may be better for workers than any pure-strategy equilibrium. I will discuss these equilbiria, as well as conditions under which they are better for workers, in Section \ref{sec:mixed_strategy}. However in a decentralized environment it is natural to expect pure strategy equilibrium to arise.   

Equilibrium will be characterized by the set $C$ of workers who screen in the first period, with the remainder pooling. Let $P^j(C)$, the \textit{observability} of $C$ to $j \not\in C$, be the probability that $j$ observes the period 1 wage of at least one worker in $C$. This is given by 
\begin{equation*}
    P^j(C) = 1 - \prod_{k \in C} (1 - \rho_{jk}).
\end{equation*}
Define $\bar{P}(C)$ to be the expected number of workers outside of $C$ who will observe a wage in $C$, which is given by
\begin{equation*}
    \bar{P}(C) = \sum_{j \in W \setminus C} P^j(C)
\end{equation*}

The following Lemma says that there are diminishing returns to observability in the size of the set $C$.\footnote{Submodularity of $\bar{P}^j$ is the key property of information sharing. In particular, we can relax the assumption that workers' observations are independent.} 

\begin{lemma}\label{lem:submodular}
$P^j(\cdot)$ is submodular, i.e. $P^j(A) + P^j(B) \geq P^j(A\cap B) + P^j(A \cup B)$. Moreover, $P^j(A) + P^j(B) > P^j(A\cap B) + P^j(A \cup B)$ if and only if the following conditions jointly hold
\begin{enumerate}
    \item $\exists k \in  A\cap B$ with $\rho_{jk} < 1$
    \item $\exists i \in A\setminus B$ with $\rho_{j,i} > 0$
    \item $\exists \ell \in B\setminus A$ with $\rho_{j,\ell} > 0$
\end{enumerate}
\end{lemma}
\begin{proof}
Using the definition of $P^j(\cdot)$, we have  $P^j(A) + P^j(B) \geq P^j(A\cap B) + P^j(A \cup B)$ iff
\begin{equation*}
    \prod_{k \in A\cap B} (1 - \rho_{jk}) + \prod_{k \in A\cup B} (1 - \rho_{jk}) \geq \prod_{k \in A} (1 - \rho_{jk}) + \prod_{k \in B} (1 - \rho_{jk})
\end{equation*}
which can be written as
\begin{align*}
    \prod_{k \in A\cap B} (1 - \rho_{jk})\cdot &\left(\prod_{k \in A\setminus B} (1 - \rho_{jk}) \cdot \prod_{k \in B\setminus A} (1 - \rho_{jk}) + 1 \right)  \\&\geq \prod_{k \in A\cap B} (1 - \rho_{jk})\cdot\left(\prod_{k \in A\setminus B} (1 - \rho_{jk}) + \prod_{k \in B\setminus A} (1 - \rho_{jk}) \right)
\end{align*}
The result follows. 
\end{proof}
A sufficient condition for strict submodularity if $P^j$ for all $j$ is that $\rho_{i,j} \in (0,1)$ for all $j$. One important implication of submodular observability is that the payoffs of the high-type firm will be supermodular in the set of screening offers it accepts. This corresponds to the intuition discussed earlier; when there are many screening workers the probability that the wage of each screening worker reveals pivotal information to other workers is lower. Additionally, when the high type firm reveals its type to many workers there are fewer workers for whom others' wages convey new information. Lemma \ref{lem:simple_supermodular} formalizes this intuition. 

Let $C$ be the set of workers screening in equilibrium in the first round, with initial offers given by $\omega_1 = \{ w^i_1\}_{i\in C}$. Let $\pi(A|C, \omega)$ be the payoff of the type $s''$ firm that accepts $w^i_1$ iff $i\in A$ (and accepts $s''$ in the second round from workers in $C\setminus A$). This is given by 
\begin{equation*}
\begin{split}
    \pi(A|C,\omega_1) &= \sum_{i\in A} \left( s'' - w_1^i \right) \\ 
    &+ \sum_{j \in W\setminus C} \left( s'' - s' + \beta\left(1- P^j(A)\right)\left( P^j(C\setminus A) (s'' - s') + (1-P^j(C\setminus A))U_2(s'')\right) \right) \\
    &+ \sum_{j \in C\setminus A} \left( \delta(s'' - s') + \beta\left(1 - P^j(A) \right)(s'' - s') \right)
\end{split}
\end{equation*}

The set function $A \mapsto \pi(A|C, \omega)$ is supermodular if $\pi(A\cup B| C,\omega) + \pi(A\cap B|C,\omega) \geq \pi(A| C,\omega) + \pi(B|C,\omega)$ for all $A,B \subseteq C$. 

\begin{lemma}\label{lem:simple_supermodular}
$A \mapsto \pi(A|C,w)$ is supermodular, strictly so when $\rho_{ij} \in (0,1) \ \forall \ i,j$.\footnote{It is easy to see from Lemma \ref{lem:submodular} that weaker conditions for supermodularity of $\pi$ can be stated. }
\end{lemma}

\begin{proof}
For simplicity I will write $\pi(X)$ rather than $\pi(X|C,\omega)$. I wish to show that $\pi(A\cup B) - \pi(B) \geq  \pi(A) - \pi(A\cap B)$. In all cases, the payoff the firm derives from workers in $A\cap B$ is unchanged. I will consider separately the payoffs of workers in $W\setminus C$, $C\setminus (A \cup B)$, $A \setminus B$ and $B\setminus A$.

Consider the payoff derived from workers in $W \setminus C$. Let $\Pi^j(X) = \prod_{k\in X}(1-\rho_{jk})$. Supermodularity of the $W\setminus C$ component of firm payoffs will follow from supermodularity of $X \mapsto 1 - P^j(X)$ and $ X \mapsto (1-P^j(X))P^j(C\setminus X) $. The former follows immediately from supermodularity of $P^j$. The expression $(1-P^j(X))P^j(C\setminus X)$ can be written as $\Pi^j(X) - \Pi^j(C)$, so supermodularity of $ X \mapsto (1-P^j(X))P^j(C\setminus X) $ follows from that of $\Pi^j$ (see Lemma \ref{lem:submodular}).

Now for the $C\setminus (A\cup B)$ component of $\pi(X)$. This follows immediately from supermodularity of $P^j$ (strictly when the conditions given in Lemma \ref{lem:submodular} are satisfied). 

Now consider the $A\setminus B$ component of payoffs. In both $\Pi(A\cup B)$ and $\Pi(A)$, these workers are successfully screening, so they generate the same payoff for the firm. Supermodularity here will follow by showing that the $A\setminus B$ component of $\pi(B)$ is less than that of $\pi(A\cap B)$. This follows from the fact that $P^j(B) \geq P^j(A)$, with strict inequality under the stated conditions for strict supermodularity. 

Finally, consider the $B\setminus A$ component of payoffs. Since these workers successfully screen in both $\pi(A\cup B)$ and $\pi(B)$, we need only show that the $B\setminus A$ component of payoffs is greater in $\pi(A\cap B)$ than in $\pi(A)$. This follows from the same argument given for the $A\setminus B$ component above.  
\end{proof}

I turn now to characterizing the wage that each screening worker offers to the high type firm. The subtlety here is that the binding incentive constraint under which type $s''$ accepts a screening wage is not the ``single deviation'' of rejecting $i$ and accepting all other $k \in C$. Supermodularity of $A \mapsto \pi(A|C,\omega)$ implies that the binding incentive constraint for the high type firm will be that associated with rejecting all offers of workers in $C$. In other words, in equilibrium the high-type firm will be indifferent between accepting all initial offers in $C$ and rejecting all initial offers in $C$. This constraint constrains the first-period initial screening offers made by workers in $C$. 

To see why this is the case, consider an equilibrium in which the set of screening workers is $C$, and these workers make initial offers of $\omega = \{w_1^i \}_{i\in C}$. $C$ and $\omega$ fully characterize on-path play in equilibrium. For each $i\in C$, let 
\begin{equation*}
    \chi^i = \argmax_{X \subseteq C\setminus i}  \pi(X|C,\omega).
\end{equation*}
$\chi^i$ need not be single valued, but when it will not cause confusion I will discuss as if it is, and refer to this set as $X^i$. In words, $X^i$ is the set of screening offers that they high-type firm will accept conditional on rejecting $i$'s initial offer. 

\begin{proposition}\label{prop:simple_deviation}
Assume $\rho_{ij} \in (0,1)$ for all $i,j$. In any equilibrium characterized by $C, \omega$, it must be that $\pi(C|C,\omega) = \pi(\varnothing|C,\omega)$, and $\varnothing \in \chi^i$ for all $i\in C$.
\end{proposition}
\begin{proof}
If $\rho_{ij} \in (0,1)$ for all $i,j$ then the conditions for strict supermodularity of $\pi(\cdot|C,\omega)$ are satisfied for all non-nested $A,B$. I will first show that $\varnothing \in \chi^i$ for all $i \in C$.

\noindent\textit{Claim 1}. For any $i \in C$, any $X^i \in \chi^i$, any $k \in X^i$ and any $X^k \in \chi^k$, we have $X^k \subsetneq X^i$. 

Let $X^i \neq \varnothing$, and let $k \in X^i$. By definition of $X^i$, $\pi(X^i|C,\omega) \geq \pi(X^i\cap X^k|C,\omega)$. But then Lemma \ref{lem:simple_supermodular} implies $\pi(X^i \cup X^k|C,\omega) \geq \pi(X^k|C,\omega)$. Suppose $X^i$ and $X^k$ satisfy the conditions for strict supermodularity of $\pi$. Then we have $\pi(X^i \cup X^k|C,\omega) > \pi(X^k|C,\omega)$. Consider what happens if $k$ offers a slightly higher wage. This can only make the payoff to the firm from rejecting $k$'s offer worse than before $k$'s deviation, since by rejecting the original offer of $w_1^k$ the firm would have signaled that it was type $s'$. If the type $s''$ firm accepts $k$'s offer then at worst worker $k$, and any other worker who observes $k$'s wage, will believe that the firm is type $s''$ for sure. But this is the same belief they would hold if the firm accepted $w^k_1$. Thus by raising it's wage offer slightly, worker $k$ makes only a small change to the firm's payoff from accepting the offers of workers in $X^i \cup X^k$. Since $\pi(X^i \cup X^k|C,\omega) > \pi(X^k|C,\omega)$, the firm will still accept $k$'s offer. But then $\omega$ and $C$ cannot characterize an equilibrium, since $k$ has a profitable deviation. The conclusion is that $X^i$ and $X^k$ cannot satisfy the conditions for strict supermodularity of $\pi$. However strict supermodularity holds for any non-nested $X^i,X^j$ when $\rho_{ij} \in (0,1)$ for all $i,j$. This proves Claim 1. 

\noindent\textit{Claim 2}. There exists $k \in C$ such that $X^k = \varnothing$. 

Let $i$ be an arbitrary worker with $X^i \in \chi^i$ such that $X^i \neq \varnothing$, and let $j \in X^i$. Then by Claim 1, $X^j \subsetneq X^i$ for all $X^j \in \chi^j$ (the inclusion must be strict since $j \in X^i$, and $j \not\in X^j$ by definition). Then since there are finitely many workers there must be $k \in X^i$ such that $X^k = \varnothing$. 

\textit{Claim 3}. Even without conditions on $\rho$, for any equilibrium $C,\omega$ we will have $\pi(X^i|C,\omega) =\pi(C|C,\omega)$ for all $i$.

To see this, first note that optimality of $w^i_1$ for worker $i$ implies $\pi(C|C,\omega) \leq \pi(X^i|C,\omega)$. Moreover, firm optimality implies $\pi(C|C,\omega) \geq \pi(A|C,\omega)$ for all $A\subseteq C$. Thus $\pi(C|C,\omega) = \pi(X^i|C,\omega)$. This proves Claim 3. 

By Claims 2 and 3, for any $i$ we have $\pi(X^i|C,\omega) = \pi(C|C,\omega) = \pi(X^k|C,\omega) = \pi(\varnothing|C,\omega)$, so $\varnothing \in \chi^i$, as desired. 
\end{proof}

One implication of Proposition \ref{prop:simple_deviation} is to pin down the sum of screening initial offers. Let $\bar{W}(C)$ be the sum of the equilibrium initial offers of workers in $C$, which is defined by the high-type firm's indifference condition
\begin{equation*}
\begin{split}
    |C|&\left(s'' - \frac{\bar{W}(C)}{|C|} \right) + (|W|-|C|)(s'' - s') +  \beta \left(1- \bar{P}(C) \right)U_2(s'') \\
    &= |C|\left( \delta(s'' - s') + \beta(s'' - s') \right) + (|W| - |C|)(s'' - s') + \beta\bar{P}(C)(s'' - s') + \beta(1-\bar{P}(C))U_2(s'').
\end{split}
\end{equation*}
Thus $\bar{W}(C)$ is given by
\begin{equation}\label{eq:bar_W}
    \bar{W}(C) = |C|\left( s'' - (\delta + \beta)(s'' - s') \right) - \beta \bar{P}(C)(s''-s'). 
\end{equation}

When information sharing is symmetric ($\rho_{ij} = \rho \in (0,1)$ for all $i,j$), the average screening worker initial offer is increasing in $|C|$. 

\begin{lemma}\label{lem:simple_avg_wage}
Assume $\rho_{ij} = \rho \in (0,1)$ for all $i,j$. Then $C \mapsto \bar{W}(C)/|C|$ is increasing in $|C|$. 
\end{lemma}
\begin{proof}
Recall that 
\begin{equation*}
    \frac{\bar{W}(C)}{|C|} = \left( s'' - (\delta + \beta)(s'' - s') \right) - \beta \frac{\bar{P}(C)}{C}(s''-s').
\end{equation*}
The claim follows if $C \mapsto \bar{P}(C)/|C|$ is decreasing in $|C|$. Under symmetry, $\bar{P}(C) = (|W| - |C|) \left(1-(1-\rho)^{|C|} \right) $. Then the claim is that the function $|C| \mapsto \frac{|W|-|C|}{|C|}(1- (1-\rho)^{|C|})$ is decreasing. The derivative of this function is negative iff
\begin{equation*}
    \dfrac{|W|}{|C|(|W| - |C|)} > -\ln(1-\rho)\dfrac{(1-\rho)^{|C|}}{1 - (1-\rho)^{|C|}}.
\end{equation*}
The right hand side of this expression is bounded above by $1/|C|$ (the limit as $\rho \rightarrow 0$). The inequality follows. 
\end{proof}

A necessary condition for screening to be possible with set $C$ is that $\bar{W}(C)/|C| \geq s'$, otherwise the low type firm will want to accept initial screening offers as well. This holds iff
\begin{equation}\label{eq:simple_Wbound}
    (1 - \delta - \beta) \geq \beta\frac{\bar{P}(C)}{|C|}.
\end{equation}
The right hand side of (\ref{eq:simple_Wbound}) non-negative, so it must be that $\beta + \delta \leq 1$. Under symmetric information sharing $C \mapsto \bar{P}(C)/|C|$ is decreasing (as shown in the proof of Lemma \ref{lem:simple_avg_wage}), so (\ref{eq:simple_Wbound}) holds for $C$ above some cut-off. 

A worker who is supposed to screen in equilibrium can always deviate to pooling (knowing that no other workers will be able to screen, as we will see below). 
\begin{lemma}\label{lem:simple_screendeviate1}
Let $C$ be the equilibrium set of screening workers. If $i\in C$ deviates and makes an initial offer of $s'$ then both firms accept. 
\end{lemma}
\begin{proof}
This is obvious for the type $s'$ firm. Consider the type $s''$ firm. If the firm accepts $s'$, any other workers for whom $i$'s wage is pivotal information will infer that the firm is type $s'$, which is the best outcome for the firm. Thus the only reason to reject $s'$ is to convince $i$ to pool rather than screen in the second period. But, as shown in the proof of Lemma \ref{lem:simple_pool}, this is not sufficient to justify delay in the first period, so the high-type firm accepts the initial offer of $s'$. 
\end{proof}

If a screening worker deviates to pooling, they know that the firm will reject the initial offers of all other workers in $C$. This is important because it implies that the incentive of workers to screen depends on $C$ and $\rho$ only through $\bar{W}(C)$
\begin{lemma}\label{lem:simple_screendeviate2}
Assume $\rho_{ij} = \rho \in (0,1)$ for all $i,j$. In a symmetric equilibrium, if a worker in $C$ makes an initial offer of $s'$, the high-type firm will reject the initial offers of all other workers in $C$.
\end{lemma}
\begin{proof}
Let $C$ be the equilibrium screening set. Since the equilibrium is symmetric by assumption, each worker in $C$ makes an initial offer of $\bar{W}(C)/|C|$. Suppose $i$ deviates and offers $s'$ initially. Compare the payoff to the high-type firm from continuing to accept the initial offers from some set $A \subseteq C\setminus i$ to $\pi\left(A|C\setminus i, \bar{W}(C\setminus i)/(|C| - 1)\right)$ (i.e. the payoff of accepting only initial offers from $A$ in a symmetric \textit{equilibrium} in which the screening set is $C \setminus i$). The only difference between these two scenarios, aside from the different initial offers in $A$, is the payoff obtained from workers in $W \setminus C$ who observe the wage of $i$ and no wages from $A$. When such workers believe that $C$ is pooling, they will infer nothing after observing that $i$'s wage was $s'$. On the other hand, when these workers believe that $s'$ is screening, they will infer that the firm is type $s'$. Thus the firm has more to lose by accepting the initial offers in $A$ in the latter case than in the former (strictly more iff uninformed workers screen in the second period). Moreover, since $X \mapsto \bar{W}(X)/|X|$ is decreasing, $\frac{|C|-1}{|C|} \bar{W}(C) > \bar{W}(C\setminus i)$. When the equilibrium screening set is $C\setminus i$, the high-type firm prefers rejecting all initial offers to accepting only those in $A$. Therefore the high-type firm will also prefer to reject all offers in $A$ when firm $i$ deviates to pooling. 
\end{proof}

In summary, the equilibrium incentive constraints in a symmetric equilibrium with symmetric information sharing are as follows:
\begin{itemize}
    \item The high type firm must be indifferent between accepting all initial offers from workers in $C$ and rejecting all offers from workers in $C$. Thus the sum of these offers must be $\bar{W}(C)$.
    \item The low type firm will accept an offer if and only if it is $s'$. This implies that we must have $\bar{W}(C)/|C| > s'$, the condition for which is (\ref{eq:simple_Wbound}). 
    \item Workers in $C$ must be indifferent between screening and deviating to pooling. If they deviate they know that the firm will reject all other screening initial offers, so they will receive no information from other workers. 
\end{itemize}

I turn now characterizing the relevant deviation for pooling workers. In particular, I identify the highest period 1 initial wage offer that the high-type firm will accept from a pooling worker. Assume that there is an equilibrium with screening set $C$, and suppose $i\in W\setminus C$ deviates by making a wage offer greater than $s'$. Let $\Tilde{w}(C,i)$ be the highest wage offer that $i$ can make such that the high-type firm is indifferent between accepting all initial offers in $C\cup i$ and rejecting all such offers. $\bar{W}(C) + \Tilde{w}(C,i)$ may differ from $\bar{W}(C\cup i)$ because of the inferences made by workers when the firm rejects $i$'s initial offer; when $i$ is expected to screen in equilibrium, other workers who observe a wage of $s'$ for $i$ will infer that the firm is type $s'$, while when $i$ deviates from pooling to screening and is rejected, other workers make no inference upon observing a wage of $s'$ for $i$. As a result $\Tilde{w}(C,i)$ is given by
\begin{equation*}
    \Tilde{w}(C,i) = s'' - (\delta + \beta)(s''-s') - \beta\left( \bar{P}(C\cup i) - \bar{P}(C) \right)U_2(s'')
\end{equation*}

\begin{lemma}\label{lem:simple_pooldeviate}
Assume $\rho_{i,j} \in (0,1)$ for all $i,j$. In an equilibrium with screening set $C$, if a pooling worker $i$ makes an initial offer of $\hat{w}> s'$, the type $s''$ firm will accept iff $\hat{w} \leq \Tilde{w}(C,i)$. 
\end{lemma}
\begin{proof}
Suppose that the type $s''$ firm rejects $i$'s offer. Then $i$ will believe that the firm is type $s'$ as long as it does not observe evidence to the contrary. Since the second period payoff to the firm from a worker who believes it is type $s'$ is weakly higher when workers maintain their prior, the incentive of the firm to screen with the workers in $C$ is lower than in equilibrium. By Lemma \ref{prop:simple_deviation}, for all $A \subseteq C$, on path the firm weakly prefers rejecting all initial offers in $A$ to accepting all initial offers in $A$, and only these initial offers. Thus following its rejection of $i$'s initial offer, the firm will also weakly prefer rejecting all initial offers in $A$ to accepting all initial offers in $A$, for all $A \subseteq C$. 

By the definition of $\Tilde{w}(C,i)$, the firm is indifferent between accepting all initial offers in $C \cup i$ and rejecting all such offers. Combined with the preceding paragraph, this implies that if $i$ makes an initial offer of $\Tilde{w}(C,i)$, the type $s''$ firm will accept. In particular, the firm will prefer accepting all initial offers in $C\cup i$ to rejecting all initial offers in $C\cup i$. (Moreover, given acceptance of $i$'s offer, the high-type firm is more willing to be screened by other workers than it was with no deviation by $i$.)

It remains to show that no higher offer by worker $i$ will be accepted. This would only be the case if and only if there was some set $A \subset C$ such that the firm strictly preferred accepting only the initial offers in $A\cup i$ to accepting all initial offers in $C \cup i$. The existence of such a set $A$ is ruled out however by submodularity of $P^j$ (Lemma \ref{lem:submodular}), using a similar argument to that given for Lemma \ref{lem:simple_supermodular}. (It does not follow directly from Lemma \ref{lem:simple_supermodular} because we are comparing a case in which $i$ pools to one in which it screens. However the difference in the argument is minor.)
\end{proof}

A full characterization of (pure-strategy) equilibrium requires characterizing the set of screening sets $C$ for which the pooling worker does not want to deviate as identified in Lemma \ref{lem:simple_pooldeviate}. Unfortunately this set does not have a convenient characterization in general. Rather than focus on this question here, I will return to the question of pooling worker incentives in Section \ref{sec:worker_heterogeneity}.

\subsection{Mixed strategies}\label{sec:mixed_strategy}
It is well known from the literature on dynamic mechanism design with limited commitment that the uninformed player, in this case the worker, may wish to induce the informed player to randomize in order to spread information revelation over time and reduce information rents.\footnote{See for example \cite{doval2019optimal}.} In the current setting such randomization would take the form of partial screening in the first period; the high type firm would reject an initial offer of $w_1 > s'$ with some interior probability. This will be beneficial to workers when $p$, the probability of a high-type firm, is sufficiently high. 

Let $C$ be the set of screening workers, and suppose that the high-type firm rejects the screening offers of all initial workers with probability $\sigma$. Let $w_1 > s'$ be the initial first-period offer made by workers in $C$. Since the low-type firm always strictly prefers to reject such an offer, acceptance of $w_1$ continues to reveal the firm as high-type. Randomization may however affect the behavior of workers who observe a worker in $C$ with a wage of $s'$. If, following such an observation, workers continue to make make initial second-period offers of $s'$ then information rents are unaffected and the high-type firm's randomization makes workers worse off. If, however, these workers screen in the second period instead of offering $s'$ then the payoff to the high-type firm of rejecting $w_1$ is lower. This reduces the information rents that must be paid to the high-type firm. 

Recall that the worker will prefer screening to pooling in the second period if their belief $q$ is such that $qs'' \geq s'$. Let $p^* = s'/s''$. Suppose that the high-type firm is indifferent between accepting all initial offers in $C$ and rejecting all initial offers in $C$. For randomization by the high-type firm to induce screening in the second period it must be that $p > p^*$. In particular, if with probability $\sigma(p) = \frac{p^*}{1-p^*}\frac{1-p}{p}$ the firm rejects all first-period initial screening offers then the belief of workers who observe a wage of $s'$ for a worker in $C$ will be exactly $p^*$. Assume that workers screen in this case. 

Let $\Tilde{W}(C)$ be the sum of first period initial offers of workers in $C$ such that the high-type firm is indifferent between accepting all initial offers in $C$ and rejecting all initial offers in $C$. Since rejection is now less profitable for the firm, relative to the equilibrium with deterministic acceptance, we will have $\Tilde{W}(C) > \bar{W}(C)$. In particular
\begin{equation*}
    \Tilde{W}(C) = |C|\left( s'' - \delta(s''-s') \right) - \beta \bar{P}(C)\delta (s''-s'),
\end{equation*}
so $\frac{\Tilde{W}(C)}{|C|} - \frac{\bar{W}(C)}{|C|} = \beta(s''-s') + \beta \frac{\bar{P}(C)}{|C|}(1-\delta)(s''-s')$.

Randomization allows screening workers to make higher initial offers. There are, however, two losses to worker welfare that arise as a result of firm randomization. First, the screening workers initial offer is sometimes rejected by the high-type firm. Second, all workers receive less precise information about the firm's type. Notice that $\sigma(p)$ is decreasing in $p$, and $\sigma(1) = 0$. When $p$ is high the high-type firm need only reject initial offers with a low probability to induce screening in the second period. On the other hand, $\Tilde{W}(C)$ is independent of $p$, so the gains from reducing information rents are constant. Therefore when $p$ is high enough even a small degree of randomization can induce a discrete jump in the payoff of screening workers. In particular, the payoff of a screening worker is given by
\begin{equation*}
    p(1-\sigma)\left( \frac{\Tilde{W}(C)}{|C|} + \beta s'' \right) + \left((1-p) + \sigma p \right)\left(\sigma s' + \beta s' \right).
\end{equation*}

Fixing $C$, pooling workers always suffer from randomization, but their losses will be more than offset by the gains to screening workers when $p$ is high. More importantly, it may be that without randomization screening is not possible, as information rents are too high. In this case the unique pure-strategy equilibrium involves no screening in the first period. If this is the case then randomization which allows some screening to occur can improve payoffs for \textit{all} workers.

\section{Worker heterogeneity}\label{sec:worker_heterogeneity}

Thus far I have assumed that workers are identical. This poses a problem for characterizing equilibrium. As Section \ref{sec:simple} explains, the set $C$ of screening workers fully characterizes the incentive constraints of high and low-type firms, as well as those of the screening workers (assuming symmetric screening offers). It is important for each worker to know the identity of every worker in $C$, as this determines their inferences both on and off path. However it is not obvious how asymmetries between identical workers, with some screening and others pooling, as well as the knowledge of which workers are doing which activity, would arise in a decentralized environment. 

In many settings workers may not in fact be identical. They may differ in their patience, degree of risk aversion, disutility of labor, outside option, etc. Heterogeneity among workers, particularly in observable characteristics, will facilitate coordination on the set of screening workers. Moreover, many sources of heterogeneity will make some workers more willing to engage in screening than others. In this section I characterize equilibrium in such settings, and study welfare comparative statics and implications for policy interventions. As might be expected, increasing information sharing between workers (i.e. increasing $\rho$) will adversely affect the welfare of screening workers by increasing in the information rents that must be paid to screen high-type firms. Moreover, when information sharing is too high no screening will be possible in equilibrium, reducing the welfare of pooling workers as well. The main insight however is that the positive externality imposed by one screening worker on another, embodied by the fact that $C \mapsto \bar{W}(C)/|C|$ is increasing in $|C|$, can be leveraged to offset the negative effects of increasing $\rho$. Thus policies aimed at increasing worker welfare should not only promote transparency, but also provide additional incentives for screening. 

Assume $W = P \cup S$, where $P$ is the set of workers who do not screen when their belief is $p$, their prior, even if all other workers are screening. These workers make a wage offer of $s'$ in both rounds of the first period. In the second period they will screen iff their belief that the firm is type $s''$ is greater than some cut-off $q' > p$. If they believe the firm to be type $s''$ for sure then they demand $s''$ in the second period. I am not assuming that these workers are ``behavioral'', in the sense of mechanically playing a given strategy, only that the threshold belief above which they find it optimal to screen is higher than that of workers in $S$. There are a number of reasons that this might be the case:  higher discounting between rounds, worse outside option (if within-round discounting is reinterpreted as a separation probability), greater risk aversion, cultural norms. I will not explicitly model any of these reasons here. Rather, I will explore the implications of such preferences on equilibrium outcomes, and discuss welfare comparative statics that apply to a broad class of preferences for workers in $P$. 

Workers in $S$ have payoffs as described in the previous section. I will focus on symmetric equilibria in which all workers in $S$ follow the same strategy. Note that if there is no equilibrium in which all workers in $S$ screen then there is no equilibrium in which any workers screen. Thus restricting attention to symmetric equilibria does not artificially disadvantage the workers. 

\begin{proposition}\label{prop:simple_eq}
Assume symmetric information sharing. There exists an integer $n(\rho,P)$ such that in the unique symmetric pure-strategy equilibrium (when it exists)
\begin{enumerate}
    \item if $|S| \geq n(\rho,P)$ all workers in $S$ screen in first period,
    \item if $|S| < n(\rho,P)$ no workers screen in the first period, 
\end{enumerate}
where $n(\rho,P)$ is increasing in both arguments. Moreover, a symmetric pure strategy equilibrium always exists if  $|S| \geq n(\rho,P)$, and exists when $|S| < n(\rho,P)$ if \textit{i}) uninformed workers do not screen in the second period (so $ps'' \leq ps'$), and/or \textit{ii}) $\bar{P}(S)/|S| \leq \delta\bar{P}(1) \equiv (|W|-1)\rho$.\footnote{Condition $ii.$ is a weak requirement; since $S \mapsto \bar{P}(S)/|S|$ is decreasing, it holds as long as $\delta$ and/or $|S|$ is not too low. As explained in the proof, this condition guarantees that if there is no equilibrium with screening set $S$ then there is an equilibrium with no screening. It is unnecessary if workers observe the entire negotiation of other workers, rather than just their realized wage. In that case a symmetric pure-strategy equilibrium always exists.}
\end{proposition}
\begin{proof}
The existence of $n$, and its monotonicity properties, are immediate consequences of by Lemmas \ref{lem:simple_avg_wage}, \ref{lem:simple_screendeviate1}, and \ref{lem:simple_screendeviate2}. Existence follows from the results of Section \ref{sec:simple}. If workers in $S$ are not willing to screen when the screening set is $|S|$ then they will be unwilling to do so for any smaller screening set, by Lemmas \ref{lem:simple_avg_wage}, \ref{lem:simple_screendeviate1}, and \ref{lem:simple_screendeviate2}. Then the only equilibrium is for no workers to screen. It remains to show that this is in fact an equilibrium, i.e. that no workers want to deviate to screening. The reason this does not follow immediately from the fact that $C \mapsto \bar{W}(C)/|C|$ is decreasing is that only a worker's wage, not the full negotiation, is observed by other workers. To see this, suppose that in equilibrium no workers screen in the first period. If worker $i \in S$ deviates and makes an initial offer of $w > s'$ the high-type firm must decide whether or not to accept. The difference between this decision and the decision of whether or not to accept an \textit{equilibrium} screening offer from a single worker $i$ is the inference made by other workers if the firm rejects (i.e. the inference when workers see a wage of $s'$). If $i$ is supposed to be screening in equilibrium, other workers infer that the firm is type $s'$, and offer $s'$ in the second period. However if other workers suppose that $i$ is pooling then they will infer nothing from a wage of $s'$. If they then screen in the second period then the firm is worse off than it would be in the equilibrium screening case. Condition $ii.$ is exactly the condition which guarantees that the highest $w$ that such a pooling worker could offer when deviating to screening is less than $\bar{W}(S)/|S|$. Thus if screening cannot be supported when all workers in $S$ screen then it cannot be a profitable deviation for a pooling worker when all other workers pool.\footnote{Condition $ii.$ is not necessary for this to be the case: it says nothing about the relative probability of the high type, which also matters for the decision to screen.}
\end{proof}

This characterization has important welfare and policy implications. The following lemma follows immediately Proposition \ref{prop:simple_eq}.  

\begin{lemma}\label{lem:simple_rhowelfare}
Pooling-worker welfare is non-monotone in $\rho$: it increases in $\rho$ as long as $|S| \geq n(\rho,P)$ and is minimized when $n(\rho,P) > |S|$. Screening-worker welfare is decreasing in $\rho$.
\end{lemma}

Worker welfare is increasing in $|S|$, fixing $|P|$. 
\begin{lemma}\label{lem:simple_Swelfare}
Fix $|P|$. The welfare of all workers is increasing in $|S|$. In particular, each screening worker receives higher wages from the high-type firm, and pooling workers receive more information, the higher is $|S|$.
\end{lemma}
\begin{proof}
This follows from showing that $S\mapsto \frac{|P|}{|S|}(1-(1-\rho)^{|S|})$ is decreasing. The derivative with respect to $|S|$ is negative iff 
\begin{equation*}
 -\frac{1}{|S|^2}(1-(1-\rho)^{|S|}) - |S|\ln(1-\rho) (1-\rho)^{|S|} < 0.
\end{equation*}
The limit of the LHS as $\rho \rightarrow 1$ is $-1/|S|^2 < 0$ and the limit as $\rho \rightarrow 0$ is $0$. The derivative of the LHS with respect to $\rho$ is 
\begin{equation*}
   - \frac{1}{|S|} (1-\rho)^{(|S|-1)} + |S|(1-\rho)^{(|S|-1)} + |S|^2 \ln(1-\rho)(1-\rho)^{(|S|-1)}. 
\end{equation*}
The sign of this is the same as that of 
\begin{equation*}
    -\frac{1}{|S|} + |S| + |S|^2\ln(1-\rho)
\end{equation*}
which is (weakly) negative for any $|S| \geq 1$. Thus the LHS term is decreasing in $\rho$. Since the limit of the LHS term as $\rho \rightarrow 0$ is 0, this implies that the derivative of $\frac{|P|}{|S|}(1-(1-\rho)^{|S|})$ with respect to $S$ is negative. 
\end{proof}

Worker welfare is decreasing in $|P|$, fixing $|S|$.

\begin{lemma}
Fix $|S|$. The welfare of all workers is decreasing in $|P|$. In particular, each pooling worker receives lower wages from the high-type firm. Pooling worker's information is unchanged as long as $n(p) \geq |S|$, and they receive no information otherwise. 
\end{lemma}

The main policy take-away is that increasing $\rho$ alone may have adverse effects. These can be mitigated by measures which encourage workers in $P$ to engage in screening. The specific ways in which this can be done, and the potential for welfare gains, will depend on the reasons for which workers in $P$ are reluctant to screen. Nonetheless, the positive externalities of screening, both for pooling and screening workers, justify interventions along these lines. 

\section{Application: discrimination}\label{sec:paysecrecy}

On of the primary motives for increasing pay transparency is to reduce discrimination. The logic generally given is that greater transparency makes it easier for workers to identify discriminatory compensation patterns. This should in turn prevent firms from discriminating. In order to more fully address the benefits of pay-secrecy policies, it is important therefore to incorporate discrimination into the model. I will be particularly interested in how various anti-discrimination policies affect equilibrium outcomes, and how these policies interact with transparency. 

The legal background of pay discrimination policies will form the basis for incorporating discrimination into the model. Pay discrimination occurs along many lines, including race, gender, and sexual orientation. In this discussion I will focus on laws regarding sex-based pay discrimination, although the model will apply equally well to other forms of discrimination. The primary legal basis for the prohibition of sex-based discrimination is the Equal Pay Act of 1963. This law ``prohibits employers from discriminating among employees on the basis of sex by paying higher wages to employees of the opposite sex for `equal work on jobs the performance of which requires equal skill, effort, and responsibility, and which are performed under similar working conditions.' '' Belfi v. Prendergast, 191 F.3d 129, 135 (2d Cir.1999) (quoting 29 U.S.C. \S  206(d)(1)). The law allows for four defenses for observed pay discrepancies: the employer must show that the discrepancies are the result of a seniority system, a merit system, a system that measures quality or quantity of production, or ``any factor other than sex.'' The interpretation of this latter catch-all defense has naturally been the subject of great controversy. One such factor which has at times been cited in defense of pay differences is the so-called ``salary negotiation defense''. This is the argument that the difference in pay arises from differences in negotiation tactics. Essentially, the employer may assert that a man is paid more than a woman for the same work because he bargained hard, while she did not. The legal precedent for this type of defense is mixed. The argument was rejected in Dreves v. Hudson Group Retail LLC, 2013 WL 2634429, **8-9 (D. Vt. June 12, 2013). However versions of the salary negotiation defense were accepted in Muriel v. SCI Arizona Funeral Services, Inc., 2015 WL 6591778, *3 (D. Ariz. Oct. 30, 2015) and  Horner v. Mary Institute, 613 F.2d 706, 714 (8th Cir. 1980). The use of the salary negotiation defense raises the question of how well an anti-discrimination law which admits such a defense does in achieving it's goals of eliminating pay discrimination. 

\subsection{Verifiable discrimination}

I begin by assuming that the anti-discrimination law admits a salary negotiation defense. This in turn defines the \textit{prima facie} burden of proof for an employee who wishes to bring a suit against their employer. Let $Y\subset W$ be the set of workers which is protected by the anti-discrimination law. I define verifiable discrimination as that which can not be justified via the salary negotiation defense. 

\begin{definition}
There is \textbf{verifiable discrimination} against $i\in Y$ if $i\in Y$ makes an offer of $w \in (s',s'']$ which is rejected, and $i$ observes $j\not\in Y$ with $w^j > s'$.\footnote{As will be shown below, the results will establish the irrelevance of penalties for verifiable discrimination. These negative results continue to hold with the following narrower definition of verifiable discrimination: $i$ makes an offer of $w^i$ which is rejected, and $i$ observes $j\not\in Y$ with a wage $w^j \geq w^i$. This definition corresponds more directly to the intuitive idea of the salary negotiation defense. However I focus on the more permissive notion because to highlight the fact that even if the bar for a successful salary negotiation defense is high, the anti-discrimination policy will not achieve its goal.}
\end{definition}

I will assume that the firm pays a penalty $\ell$ for cases of verifiable discrimination, a fraction $\alpha \in [0,1]$ of which goes directly to the injured worker. A worker will bring a discrimination case against the firm if and only if they have evidence of verifiable discrimination. This rules out cases in which the worker brings the suit without making a \textit{prima facie} case. Such cases are unlikely to survive an initial motion to dismiss, and even if they do, the worker runs the risk that in fact no discrimination has occurred, in which case they will be responsible for the legal costs associated with brining the suit. 

The general intuition for why the penalty for verifiable discrimination may help $Y$ workers obtain higher wages is the following. Suppose $i\in Y$ is receiving a relatively low wage, either because they are pooling or because they are screening, but at a low wage. The larger is $\ell$, the more willing the firm should be to accept should $i$ deviate and make a higher wage offer, provided this offer is still below the wage obtained by some worker $j \in C\setminus Y$. This is because rejection of $i$'s deviation would form the basis for a verifiable case of discrimination. Greater transparency should reinforce this effect, since it means that $i$ will be more likely to observe $j$'s wage. This intuition seems to be behind the arguments in favor of both large penalties for discrimination and greater pay transparency. 

The intuition outlined above however is incomplete, as it implicitly takes as given the probability that $i$ will observe a high wage from $j$. Both $j$'s offer and the firm's decision whether to accept or reject are endogenous. Faced with $i$'s deviation, the firm considers not only whether or not to accept $i$'s offer, but also which offers in $C$ to accept. It turns out that this endogeneity severely limits the efficacy of penalties and transparency for reducing pay discrimination. In many cases, the anti-discrimination penalty will in fact have no effect on equilibrium outcomes. The reason again has to do with supermodularity of the firm's payoffs (\Cref{lem:simple_supermodular}). In the language of this paper, the intuition for why transparency and discrimination penalties help prevent discrimination can be rephrased as follows: if a pooling $Y$ worker tries to deviate to screening, the firm will have a greater incentive to accept the offer when $\rho$ and $\ell$ are high. This makes it harder to sustain equilibria in which $Y$ workers do not screen. The problem with this reasoning is \Cref{lem:simple_pooldeviate}. If the firm rejects the initial offer of a worker $i \in Y$ who deviates to screening then it will reject all initial screening offers when $\ell = 0$, and all the more so when $\ell >0$. Thus there will be no verifiable discrimination, and the penalty $\ell$ is irrelevant. This conclusion is robust across a number of different discrimination settings, a few of which I will discuss below. 

\subsection{Perception of negotiators}

There is significant empirical evidence that racial and gender pay gaps arise in part from the perceptions of workers who engage in negotiations. \cite{bowles2007social} find that women who attempt to negotiate for higher salaries are perceived as demanding ``less nice'', particularly be male evaluators. \cite{hernandez2019bargaining} show that Black workers face similar perception effects when dealing with White evaluators. 

Formally, we can model this perception effect as a utility cost for $Y$ workers who attempt to negotiate. I will assume that this cost takes the form of a fixed cost $c$ paid by a worker who proposes a wage above $s'$ (the same results will hold if instead the cost is proportional to the wage proposal). This cost is also paid if the worker brings a lawsuit against the company. I assume that this cost is not so large that it completely outweighs the benefit of learning the firm's type. To be precise, $s'' - c > s'$, so a $Y$ worker is willing to demand a high wage in the second period if they know that the firm is the high type.

Penalties for verifiable discrimination which are paid to the injured worker could in theory create perverse incentives for a worker $i \in Y$ who pools in the first period. Suppose the firm is high-type, and $i$ and observes the wage of $j\in C\setminus Y$; it could be that $i$ prefers in the second period to hide the fact that they know the firm's type, and instead make an offer that they knows the firm will reject. This would benefit $i$ if the payoff they receive from winning the subsequent lawsuit is greater then the wage that they could demand if they know the firm's type. 

I will assume that cross-period evidence cannot be used in discrimination cases. This means that a worker $i$ who observes $j \not\in Y$ with a period 1 wage greater than $s'$ will simply demand a wage of $s''$ in the second period. Ruling out the extortionary behavior described in the previous paragraph is the only effect of this cross-period evidence assumption.  

Aside from this negotiation cost and the penalty for discrimination, the model is unchanged. I will consider fixed negotiation costs $c$ throughout, and so will suppress dependence on $c$ in the notation. For any $\ell \geq 0$, let $\pi(A|C,\omega,\ell)$ be defined as before, except that the penalty $\ell$ for verifiable discrimination is incorporated. The conclusions of \Cref{prop:simple_deviation} continue to hold, regardless of the penalty for verifiable discrimination.

\begin{lemma}\label{lem:perception_wages}
Assume $\rho_{ij} \in (0,1)$ for all $i,j$ and $\ell \geq 0$; and let $c \geq 0$ be the negotiating cost. In any equilibrium characterized by $C, \omega$, it must be that $\pi(C|C,\omega,\ell) = \pi(\varnothing|C,\omega,\ell)$, and $\varnothing \in \chi^i$ for all $i\in C$.
\end{lemma}
\begin{proof}
The proof proceeds by the same steps as that of \Cref{prop:simple_deviation}. The additional penalty for verifiable discrimination only reinforces supermodularity, as it becomes more costly to reject a screening offer from $i\in Y$ when more workers $i \in C\setminus Y$ are being accepted.  
\end{proof}

As a result, it is immediate that \Cref{lem:simple_avg_wage}, \Cref{lem:simple_screendeviate1}, \Cref{lem:simple_screendeviate2} continue to hold as well. On direction of \Cref{lem:simple_pooldeviate} is also immediate: if a pooling worker deviates to screening then the high-type firm will accept their offer if it is below $\Tilde{w}(C,i)$. The conclusion that higher offers will be rejected will turn out to hold, but this is not immediate since it could be that the firm accepts a higher offer from $i\in Y$ in order to avoid a lawsuit. 

The fact that $Y$ workers pay a perception cost for negotiation will make them less willing to attempt to screen, in the absence of penalties for verifiable discrimination. The intuition outlined in the previous section suggests that penalties for verifiable discrimination may help $Y$ workers by allowing them to obtain a higher wage, and thus offset the perception cost incurred by bargaining. The endogeneity of observations, however, completely eliminates this potential benefit.  

\begin{proposition}\label{prop:perception_discrimination}
Let $\ell \geq 0$, $c \geq 0$ and $\alpha \in [0,1]$ be arbitrary. Assume $\rho_{ij} = \rho \in (0,1)$ for all $i,j$ and symmetric screening offers. Then the set of equilibrium outcomes $C,\omega$ is independent of $\ell$.
\end{proposition}
\begin{proof}
Assume for simplicity of notation that $\alpha = 1$; the proof is identical for any $\alpha \in [0,1]$. First, let $C,\omega$ characterize an equilibrium under $\ell = 0$. I wish to show that they also characterize an equilibrium under $\ell' > 0$. To do this, I verify that there will be no profitable deviations for any workers or for the firm. Assume that $C\setminus Y \neq \varnothing$, otherwise it is obvious that the penalty has no effect. 

Every worker will either screen or pool, so there are potentially four types of workers whose deviations we must consider when $\ell >0$.

\textit{Worker} $i \in C \cap Y$. Suppose first that such a worker deviates to pooling. Under both $\ell=0$ and $\ell'$, the firm will then reject all initial offers in $C$ (by \Cref{lem:simple_screendeviate2}). Thus the payoff to $i$ from this deviation is the same in both cases, and therefore not profitable under $\ell'$. The same is true if $i$ deviates by making a higher offer. 

\textit{Worker} $i \in C\setminus Y$. Again, if $i$ deviates to pooling or makes a higher offer all other offers in $C$ will be rejected, so the payoff to $i$ from deviating is independent of $\ell$.

\textit{Worker} $i \in Y\setminus C$. This is the most interesting direction. Let $\Tilde{w}$ be the highest wage that $i$ could offer under $\ell = 0$ which would be accepted by the high-type firm. By \Cref{lem:simple_screendeviate2}, under $\ell =0$ the firm will then be indifferent between accepting all initial offers in $C\cup i$ and rejecting all such offers, should $i$ offer $\Tilde{w}$. Suppose that under $\ell'$, $i$ deviates by making an offer $w \leq \Tilde{w}$. The payoff to the firm of accepting all offers in $C\cup i$ is the same as under $\ell = 0$. The payoff of rejecting $i$'s offer is weakly worse, since this opens the possibility for a lawsuit. Thus the firm will accept any such offer, and so the payoff to $i$ from deviating is the same as under $\ell = 0$. Suppose instead that $i$ offers $w > \Tilde{w}$. If the firm accepts $i$'s offer it will also accept all offers in $C$ (by supermodularity of $\pi$). If the firm rejects $i$'s offer then it will also reject all offers in $C$. This is because under $\ell = 0$ the firm at least weakly prefers rejecting all $C$ in this case (\Cref{prop:simple_deviation}), and this preference is strict for any $\ell' >0$, since otherwise there is the possibility of a lawsuit (recall the assumption that $C\setminus Y \neq \varnothing$). The payoffs to the firm of accepting all $C\cup i$ or rejecting all $C\cup i$ are independent of $\ell$, and so it will reject under $\ell'$, given that this was weakly optimal under $\ell = 0$. Thus the payoff to the worker form this deviation is weakly worse under $\ell'$ than under $\ell=0$.

\textit{Worker} $i \in (W\setminus C)\setminus Y$. The previous argument also applies to this type of worker. 

Consider now the potential deviations by the firm. The low type firm of course has no deviations, so consider a high-type firm. The only relevant deviation is rejecting some initial offers in $C$. Doing so is weakly worse under $\ell'$, since it creates the possibility of a lawsuit. Thus the firm will accept all $C$. 

Now consider an equilibrium under $\ell'$. We first need to show that the equilibrium is also characterized by $C,\omega$. In other words, we need to show that there are no strategies other than screening or pooling that could occur on path. The only other strategy that could arise is for a $Y$ worker to make high offers in the first period first round, in the hope of then encountering verifiable discrimination and winning a lawsuit. In the first period second round such a worker would then make an offer of $s'$. Call this the entrapment strategy.

Suppose there is an equilibrium in which a set $B$ of $Y$ workers plays the entrapment strategy. Let $i \in Y$ be a worker pursuing the entrapment strategy, and let $C$ be, as before, the set of screening workers. The entrapment strategy can only be profitable if $C\setminus Y \neq \varnothing$. I will show that $i$ would be better off making an offer that the high type firm would accept. For any $\ell'>0$, the firm's payoffs can be written as the sum of its operating surplus, i.e. output minus wages, and the costs arising from lawsuits. For the purposes of inference, a worker playing the entrapment strategy is the same as a pooling worker. In other words, the informational effect on operating surplus is the same whether a worker pools or entraps. For any screening set $C$ it must still be the case that $\pi(C|C,\omega,\ell) = \pi(\varnothing|C,\omega,\ell)$; the presence of entrapment workers does not affect the supermodularity of $\pi$. However $\pi(C|C,\omega,\ell)$ will now depend on $B$, since $C\setminus Y \neq \varnothing$ implies that lawsuits occur with positive probability on-path. Let $\Tilde{w}(C)$ be the highest initial wage offer from $i\in B$ that would be accepted if $\ell = 0$. In other words, $\Tilde{w}(C)$ is the highest wage that the firm would be willing to accept from $i$ if it only cared about its operating surplus. By \Cref{lem:simple_pooldeviate} and the fact that workers in $C$ are willing to screen, it must be that $\Tilde{w}(C) > s'$. 

Let $P(C\setminus Y)$ be the probability that $i$ observes a wage in $C\setminus Y$, which is exactly the probability that $i$ brings a lawsuit. Suppose that instead of making an unacceptable initial offer, $i\in B$ instead proposes $ \Tilde{w}(C) + P(C\setminus Y)\ell$. If the firm accepts, $i$ strictly prefers this deviation to the entrapment strategy, since $\Tilde{w}(C) > s'$. I claim that the firm accepts $i$'s offer, and continues to accept all initial offers in $C$. This is because $i$ screening has no effect on expected lawsuit costs arising from workers in $B\setminus i$. The firm's expected lawsuit costs will decrease by exactly $P(C\setminus Y)\ell$ if it accepts $i$'s offer. Its expected operating surplus is the same, by definition of $\Tilde{w}(C)$. Thus the firm obtains it's equilibrium payoff whether it accepts or rejects $i$'s deviation offer (to break the firm's indifference, $i$ could offer a slightly lower wage and still be better off). Thus there can be no equilibrium in which any workers entrap.

It remains to show that any equilibrium under $\ell'$ characterized by $C,\omega$ will also be an equilibrium under $\ell = 0$. By \Cref{lem:perception_wages}, it must be that $\pi(C|C,\omega,\ell') = \pi(\varnothing|C,\omega,\ell')$. Notice that $\pi(C|C,\omega,\ell') = \pi(C|C,\omega,0)$ and $\pi(\varnothing|C,\omega,\ell') = \pi(\varnothing|C,\omega,0)$, since no verifiable discrimination occurs in either case. If a screening worker deviates, either to pooling or by offering a higher wage, then all offers in $C$ are rejected, by \Cref{lem:simple_screendeviate1} and \Cref{lem:simple_screendeviate2}. If this is at least weakly worse for the deviator than what would have occurred under $\ell'$, and so cannot be profitable. We have already shown that the maximum possible wage that a pooling worker could obtain by deviating to screening is the same under $\ell$ and $\ell'$, so there is no profitable deviation of this type as well. Since $\pi(\varnothing|C,\omega,0) = \pi(C|C,\omega,0)$, there is no profitable deviation for the firm. 
\end{proof}

\subsection{Other forms of discrimination}

Discrimination may occur for a number of reasons. Two prominent forms of discrimination are statistical and taste-based. Statistical discrimination occurs if worker's have different abilities, and the employer believes that $Y$ workers are more likely to have low ability. This belief may well be unfounded, yet it will still influence hiring and pay decisions. Similarly, taste-based discrimination occurs when the employer has an intrinsic dislike for some type of employee. 

Suppose that there is tasted based discrimination, which takes the form of a utility cost $C$ payed by the firm when it employs a type $Y$ worker. This can be interpreted with as the psychological cost associated which captures the firm's bias. For simplicity, assume that this cost is not discounted; the firm pays the cost if it employs a $Y$ worker in a given period, whether they reach agreement in the first or second round of negotiation. As long as $s'' - c > s'$ the conclusion of \Cref{prop:perception_discrimination} continues to hold; penalties for discrimination will not change the set of equilibrium outcomes. The proof is unchanged. 

\begin{corollary}
Assume there is taste-based discrimination. Let $\ell \geq 0$, $c \geq 0$ and $\alpha \in [0,1]$ be arbitrary. Assume $\rho_{ij} = \rho \in (0,1)$ for all $i,j$ and symmetric screening offers. Then the set of equilibrium outcomes $C,\omega$ is independent of $\ell$.
\end{corollary}

The analysis of statistical discrimination is slightly more complicated. Assume that statistical discrimination takes the following form. Workers can be either high ability ($g$) or low ability ($b$). Ability is irrelevant at the low type firm. At the type firm a high ability worker produces output $s''$, while a low ability worker produces output $s'' - c > s'$. Neither workers nor the firm know their ability at the start of the first period, but both learn it by the end of the first period (independent of the firm's type). The firm discriminates because it believes that $Y$ workers are low ability with higher probability than non-$Y$ workers. Without loss of generality, assume that it believes non-$Y$ workers are always high ability. This belief may well be incorrect.

The difference between statistical discrimination and taste-based discrimination, as described above, is that the expected difference in output between $Y$ and non-$Y$ workers is discounted by the firm between negotiation rounds. This complicates the analysis. However under some conditions \Cref{prop:perception_discrimination} continues to hold. For any $c \geq 0$ and screening set $C$, define $\bar{W}(C|c)$ to be the wage sum such that the firm is indifferent between accepting all screening offers and rejecting all screening offers when each worker makes an offer of $\bar{W}(C|c)/|C|$. \Cref{eq:bar_W} defines $\bar{W}(C|0)$. When $c > 0$ and $C\cap Y \neq \varnothing$ then $\bar{W}(C|c) < \bar{W}(C|0)$, since the firm loses less by delaying agreement with $Y$ workers when $c > 0$ then when $c = 0$. The key property that must be preserved for \Cref{prop:perception_discrimination} to hold is that $C \mapsto \bar{W}(C|c)/|C|$ should be strictly increasing, in the set inclusion order. The proof is unchanged from that of \Cref{prop:perception_discrimination}. 

\begin{corollary}
Assume there is statistical discrimination and $C \mapsto \bar{W}(C|c)/|C|$ is strictly increasing (in the set inclusion order). Let $\ell \geq 0$, $c \geq 0$ and $\alpha \in [0,1]$ be arbitrary. Assume $\rho_{ij} = \rho \in (0,1)$ for all $i,j$ and symmetric screening offers. Then the set of equilibrium outcomes $C,\omega$ is independent of $\ell$.
\end{corollary}

$C \mapsto \bar{W}(C|0)/|C|$ is strictly decreasing by \Cref{lem:simple_avg_wage}. For any $c \geq 0$ and $i \not\in Y$, $\bar{W}(C \cup i|c)/(|C|+1) > \bar{W}(C|c)/|C|$. For $i \in Y$, $\bar{W}(C \cup i|c)/(|C|+1) > \bar{W}(C|c)/|C|$ will hold provided $c$ is not too large.

\subsection{Discriminatory low-type firm}

Consider either taste-based or statistical discrimination. A discriminatory low-type firm will employ no $Y$ workers in the second period (unless there is an additional round of information sharing after the second period), since the minimum wage is $s'$. In the first period, the low-type firm will either accept any offer of $s'$ made by a $Y$ worker, and thus avoid paying any penalty for discrimination, or reject all such offers and risk paying penalties. The higher the penalty $\ell$ and the higher the information sharing parameter $\rho$, the more likely it is that the firm will avoid lawsuits by hiring type $A$ workers.

\subsection{Empirical observations}

The results on discrimination can be summarized as follows
\begin{enumerate}
    \item For low-type discriminatory firms, increasing $\ell$ and $\rho$ reduces hiring discrimination.
    \item The penalty $\ell$ has no impact on the actions of high-type discriminatory firms. Increasing $\rho$ will not have a direct impact on discrimination. 
\end{enumerate}

The second point conforms with empirical evidence, such as that discussed by \cite{babcock2009women}, that women are paid less than men for the same job, and are less likely to ask for higher pay. The gender pay gap has remained virtually the same over the past decade, despite regulations such as the 2009 Lilly Ledbetter Fair Pay Act, which increased the penalty for pay discrimination and made it easier for employees to file discrimination cases. 

Pay gaps between men and women are highest in professions such as financial services, in which there is potentially large variability in productivity across firms.\footnote{\cite{stebbins2019paygap}}. In industries in which firms are homogeneously low-type, such as food service, the gender pay gap is smaller. 

These results are also consistent with patterns of racial discrimination. Black and Hispanic/Latino workers are over-represented in low paying jobs such as food servers and porters.\footnote{\cite{solomon2019systematic}} While there are a host of reasons for this pattern, it is consistent with the observation that measures meant to prevent discrimination, through increased wage transparency and higher penalties for discrimination, are more effective in preventing hiring discrimination in low-type firms than in closing the pay gap in high-type firms.

\newpage

\bibliography{reputation}

\newpage
\appendix

\section{Alternating offers bargaining}\label{sec:alternatingoffers}

There is a set $W$ of workers, one firm, and 2 periods. The firm has private information about the per-period surplus generated by a worker, which is $s \in \{s', s''\}$, with $s'' > s' > 0$. The worker's outside option is $0$, so there are always gains from cooperation. The worker's prior belief is that the firm is of type $s''$ with probability $p$. The firm discounts between periods at rate $\beta$, and workers do not discount between periods.

In each period the firm and worker engage in alternating-offers bargaining over a wage to be paid to the worker in the first period, \`a la \cite{rubinstein1982perfect} (so there is an infinite time horizon within each period).\footnote{I assume that the worker and firm cannot commit to a second period wage in the first period. If they could then this model would be equivalent to a single worker-single firm model. I will discuss a different but related model in which workers arrive sequentially. In this case we can just assume that each worker works for only one period.} Workers and firms have a common discount factor $\delta$ in the bargaining game. After the first period, but before the second round of wage negotiations, worker $i$ observes worker $j$'s first period wage with probability $\rho_{ij}$. The solution concept is SPBNE, with some restrictions which I will introduce later on.

\subsection{Period 2}
Consider first the negotiation in the second period, which I refer to as the P2 game. In period 2 the problem separates completely across workers, so we only need to consider a single-worker negotiation. An important preliminary observation is that discounting generates single crossing in the firm's propensity to accept a given wage. In other words, high type firms are more impatient, since delay is more costly.

\begin{lemma}\label{lem:P2_singlecrossing}
In any P2 equilibrium if firm $s'$ accepts a wage proposal then so does $s''$.
\end{lemma}
\begin{proof}
Suppose that the worker has proposed $w$ in round $t$. Let $w_{t+j}$ be the lowest wage the firm can obtain at time $t+j$ for $j \in \mathbbm{N}$ by rejecting and making counter offers. Type $s$ accepts $w$ if and only if 
\begin{equation*}
    s - w_{t+j} \geq \delta^{t+j}(s - w_{t+j}) \ \ \forall \ j \in \mathbbm{N}. 
\end{equation*}
Clearly if this condition holds for $s'$ then it holds for $s''$. 
\end{proof}

Recall that in the complete information stage game in which the worker knows that the firm is of type $s$, \cite{rubinstein1982perfect} shows that there is a unique subgame perfect equilibrium in which the worker proposes a wage of $w(s) := s(1-\delta)/(1-\delta^2)$ and the firm accepts. Another general property of equilibrium strategies is that a worker will never propose a wage below $w(s')$.

\begin{lemma}\label{lem:P2_lowerbound}
No P2 equilibrium strategy involves the worker proposing a wage below $w(s')$. 
\end{lemma}

\begin{proof}
Let $\Tilde{w}$ be the infimum of the set of wages proposed by the worker following any history, in any equilibrium. A necessary condition for type $s$ to reject a proposal of $w'$ is 
\begin{equation*}
    s - w' \leq \delta(s - \delta \Tilde{w})
\end{equation*}
or equivalently,
\begin{equation*}
    w' \geq (1-\delta) s + \delta^2 \Tilde{w}.
\end{equation*}
Notice that $\Tilde{w} < (1-\delta) s + \delta^2 \Tilde{w} \Leftrightarrow \Tilde{w} < w(s)$. So if $\Tilde{w} < w(s') < w(s'')$ there exists a wage $w' > \Tilde{w}$ such that both types of firms accept any $w \leq w'$ following any history. This means that proposing any $w < w'$ at any point is dominated by proposing $w'$, which contradicts the definition of $\Tilde{w}$. 
\end{proof}

Lemma \ref{lem:P2_singlecrossing} implies that, when the worker does not know the firm's type, equilibrium in the P2 game take one of two forms, pooling or separation. In the pooling equilibrium the worker makes a wage offer in round 1 which is accepted by both firms. In the separating equilibrium the worker first makes an offer $w_h$ which is accepted by the type $s''$ firm, and then the type $s'$ firm offers $\delta w(s')$, which is accepted by the worker. The high type firm will accept $w_h$ if and only if 
\begin{equation*}
    s'' - w_h \geq \delta(s'' - \delta w(s'))
\end{equation*}
and so the highest acceptable wage that the worker can offer is $w_h = (1-\delta)s'' + \delta^2 w(s')$. This can be written as $w_h = w(s'') - \delta^2 (w(s'') - w(s'))$. The payoff to the worker of engaging in screening, when the worker attaches probability $\hat{p}$ to the firm being type $s''$, is given by 
\begin{equation*}
    V_S(\hat{p}) = \hat{p} w_h + (1-\hat{p})\delta^2 w(s').
\end{equation*}

Alternatively the worker can make an offer in the first period that is accepted by both firms. Without restriction off-path beliefs there are many equilibria that can be supported, including a wage offer of $s'$. However these equilibria are ruled out by a mild monotonicity condition.\footnote{This condition is satisfied under various refinements of SPBNE, such as the perfect sequential equilibrium solution concept of \cite{grossman1986perfect}} It turns out that under this condition the worker can do no better than if they were facing the low type firm alone.  

\vst
\noindent \textbf{Inner Weak Monotonicity.} Suppose the worker in the P2 game assigns positive probability to the firm being type $s'$. If the firm rejects an offer and makes a lower one then the worker continues to assign positive probability to type $s'$. 

\begin{lemma}\label{lem:P2_upperbound}
Assume Inner Weak Monotonicity. In any P2 equilibrium, the type $s'$ firm will never accept a wage greater than $w(s')$. 
\end{lemma}

\begin{proof}
Let $\Hat{w}$ be the supremum of the set of wages accepted by type $s'$ and offered by a worker who assigns positive probability to type $s'$. Suppose the worker makes a proposal of $w$. The low type firm accepts $w$ if this is better than rejecting $w$ and proposing a lower wage. Under Weak Monotonicity, following such a rejection the worker continues to assign positive probability to type $s'$. Thus a necessary condition for $w$ to be accepted by $s'$ is.
\begin{equation*}
   s' - w \geq \delta(s' - \delta \hat{w})
\end{equation*}
or equivalently,
\begin{equation*}
    w \leq (1 - \delta)s' + \delta^2 \hat{w}.
\end{equation*}

Notice that $(1 - \delta)s' + \delta^2 \hat{w} < \hat{w} \Leftrightarrow \hat{w} > w(s')$. So if $\hat{w} > w(s')$ then there exists $\varepsilon > 0$ such that no wage greater than $\hat{w} - \varepsilon$ is accepted by type $s'$. But this contradicts the definition of $\hat{w}$. 
\end{proof}

\begin{corollary}\label{cor:P2_pooling}
Under Inner Weak Monotonicity, the unique NIP2 pooling equilibrium wage is $w(s')$. 
\end{corollary}

By Corollary \ref{cor:P2_pooling}, the worker screens the firms if and only if 
\begin{equation*}
    p w_1 + (1-p)\delta^2 w(s') \geq w(s') \ \ \Leftrightarrow \ \ p s'' \geq s'.
\end{equation*}

It will be important in what follows that the low type is indifferent between the outcomes of the complete information equilibrium, the pooling equilibrium, and the screening equilibrium. In the former two cases the firm immediately accepts the wage offer of $w(s')$. In the latter case the firm delays one round and then proposes the wage $\delta w(s')$, which is accepted. Indifference follows since by construction $s' - w(s') = \delta (s' - \delta w(s'))$.

\subsection{Period 1 - no information}

Consider now the game beginning in period 1, but assume that there is no information sharing between workers; $\rho_{i,j} = 0$ for all $i,j$. Call this the NIP1 (no information period 1) game. In this case the conclusion of Lemma \ref{lem:P2_upperbound} continues to hold. 

\begin{lemma}\label{lem:NIP1_upperbound}
Under Inner Weak Monotonicity, in the NIP1 game the type $s'$ firm never accepts a wage offer above $w(s')$. 
\end{lemma}
\begin{proof}
Recall that the period 2 payoff of the type $s'$ firm is the same if in the second period it's type is fully revealed, screened, or pooled. Thus the firm's second period payoff is independent of worker beliefs, so long as these assign positive probability to $s'$. The proof is then the same as that of Lemma \ref{lem:P2_upperbound}. 
\end{proof}

The separating equilibrium in the NIP1 game takes a similar form. Define $w_h^1$ by 
\begin{equation*}
     s'' - w_h^1 + \beta(s'' - w(s'')) = \delta(s'' - \delta w(s')) + \beta(s'' - w(s')).
\end{equation*}
In words, $w_h^1$ is the highest wage that type $s''$ is willing to accept in period 1 and be revealed, rather then delay by one round of negotiation to pass as type $s'$.\footnote{$w_h^1 \geq w(s')$ iff $1-\delta^2 > \beta$} Rearranging, 
\begin{equation*}
    w_h^1 = (1-\delta)s'' + \delta^2 w(s') - \beta (w(s'') - w(s')).
\end{equation*}

\begin{lemma}\label{lem:NIP1_separating}
Under Inner Weak Monotonicity, in the NIP1 game, if the equilibrium in the negotiation with a worker is separating then the worker first offers wage $w_h^1$, and the high type firm accepts. If the first offer is not accepted then the firm offers $\delta w(s')$, which is accepted. 
\end{lemma}
\begin{proof}
Since the equilibrium is not pooling, the firm's type will be revealed. Then it must be that on-path the negotiation with type $s''$ ends before that with $s'$. This is because the highest continuation payoff for $s''$ comes from convincing workers it is type $s'$, combined with the fact that waiting is more costly for type $s''$ (see Lemma \ref{lem:P2_singlecrossing}). Once the low type is revealed we are back to the complete information game, so the worker will offer $w(s')$.
\end{proof}

Without further restrictions on off-path beliefs we cannot rule out pooling equilibria with wages below $w(s')$. Such an equilibrium could be supported if the worker's believes that the firm is type $s''$ for sure flowing acceptance of a wage above the equilibrium level. Such equilbibria are ruled out under natural refinements, for example, the Perfect Sequential Equilibrium (PSE) of \cite{grossman1986perfect}. Roughly speaking, off path beliefs in PSE are guided by the following question: for a given deviation, is there a set $S$ of firm types that would benefit from deviating if the worker's belief upon observing the deviation is the restriction of the prior to $S$? If so then this should be the belief of the worker. In the case of a NIP1 pooling equilibrium with first period wage $\bar{w} < w(s')$, the deviation would be accepting a wage in $(\bar{w}, w(s')]$. If the worker believes that both types will accept such a wage then it is optimal for both types to do so. Thus no such PSE can be supported.

\begin{lemma}
Under Inner Weak Monotonicity, the unique pooling PSE in the NIP1 game is at a wage of $w(s')$. 
\end{lemma}

I will say that an equilibrium is a symmetric PSE if it is a PSE and all workers $j \neq i$ have the same beliefs about the firm's type after observing $i$'s wage. 

\subsection{Period 1 - information sharing}
The central question I want to answer is how the equilibria with and without information sharing compare. Consider now the game with information sharing beginning in the first period. A type $s'$ firm receives the same period 2 payoff regardless of worker beliefs, so long as workers continue to assign positive probability to the firm being type $s'$. The first step is to show that, as in the P2 game, the type $s'$ firm will never accept a wage above $w(s')$. This depends on a mild monotonicity condition on beliefs of workers after observing the wages of others.\footnote{If workers shared their beliefs, rather than just their wages, then this assumption would not be needed.}

\vst
\noindent \textbf{Outer Weak Monotonicity.} If worker $i$ assigns positive probability to the firm being type $s'$ after observing a wage of $w$ for worker $j$, then $i$ continues to do so after observing a wage for $j$ less than $w$.
\vst

\begin{lemma}\label{lem:P1_upperbound}
Under Inner and Outer Weak Monotonicity, the type $s'$ firm never accepts a wage offer above $w(s')$. 
\end{lemma}
\begin{proof}
Recall that the period 2 payoff of the type $s'$ firm is the same if in the second period it's type is fully revealed, screened, or pooled. Thus the firm's second period payoff is independent of worker beliefs, so long as these assign positive probability to $s'$. The proof is then the same as that of Lemma \ref{lem:P1_upperbound}, where Outer Weak Monotonicity guarantees that the beliefs of other workers do not become degenerate on $s''$. 
\end{proof}

Without further refinements the conclusion of Lemma \ref{lem:P2_lowerbound} does not hold in the first period, and so we cannot conclude that the unique pooling equilibrium wage in period 1 is $w(s')$ as in \ref{cor:P2_pooling}. This is because the type $s''$ firm wants to influence period 2 beliefs. The $s''$ is indifferent between passing as a type $s'$ and pooling with $s'$ if the $P2$ equilibrium is pooling. It strictly prefers either of these outcomes to the separating P2 equilibrium, which in turn it strictly prefers to being fully revealed as type $s''$. Depending on off-path beliefs, it may be possible to sustain pooling wages below $w(s')$ in equilibrium. Fortunately, we can do comparative statics and draw welfare conclusions without further refinements. In particular, Lemma \ref{lem:P1_upperbound} allows us to identify what the separating equilibrium in period 1 will look like.

\begin{lemma}\label{lem:P1_separating}
Under Inner and Outer Weak Monotonicity, there exist $w_1^i,w_2^i$ with $w_2^i \leq w(s')$ such that if the equilibrium in the negotiation with worker $i$ is separating then the worker first offers $w_1^i$, and the high type firm accepts. If the first offer is not accepted then the firm offers $\delta w_2^i$, which is accepted. In the unique separating symmetric PSE $w_2^i = w(s')$, and in the unique pooling symmetric PSE the wage is $w(s')$.
\end{lemma}
\begin{proof}
If the equilibrium is separating, the firm's type will be revealed. Then it must be that on-path the negotiation with type $s''$ ends before that with $s'$. This is because the highest continuation payoff for $s''$ comes from convincing workers it is type $s'$, combined with the fact that waiting is more costly for type $s''$ (see Lemma \ref{lem:P2_singlecrossing}). Lemma \ref{lem:P1_upperbound} implies $w_2^i \leq w(s')$. 

Suppose the equilibrium is pooling at wage $w$. If $w > w_2^i$ then firm $s'$ would also be willing to accept $w$ when its type has been revealed to worker $i$ in the separating equilibrium. If $w < w_2^i$ then worker $i$ could instead propose $w_2^i$ in the first period. Since the $s'$ firm was willing to accept $w_2^i$ in the separating equilibrium, it will be willing to do so in the first period. If all workers believe that both firms accept $w_2^i$ then both will, and so in a symmetric PSE this should be the off path belief. 

If $w_2^i < w(s')$ consider the deviation by a firm in the pooling equilibrium of accepting a wage $w \in (w_2^i, w(s')]$. If all workers believe that both types accept such a wage then it is optimal for both to do so, so in the symmetric PSE all workers believe that both types accept $w$. But then the worker would offer $w$ rather than $w_2^i$, so this cannot be an equilibrium. 
\end{proof}

It is not necessary to assume symmetric PSE to rule out $w_2^i = w(s')$. We can rule out equilibria in which $w_2^i < w(s')$ as long as workers who observe a wage $w \in (w_2^i,w(s')]$ do not believe that the firm is type $s''$ with probability $1$.

It need not be that $w_1^i > w_2^i$, so we cannot ignore the incentives of type $s'$ firms to accept $w_1^i$. In what follows I will use the incentive constraint of type $s''$ to pin down $w_1^i$. If it turns out that the incentive constraint of type $s'$ is violated at this wage then it will not be possible for $i$ to screen in equilibrium.

In the separating equilibrium, $w^i_1$ depends on the information sharing parameters. On the other hand, $w^i_2$ depends only on off-path beliefs. The willingness of a type $s''$ firm to reveal itself to worker $i$ depends on two characteristics of equilibrium information sharing; the number workers other than $i$ that are not screening the firm in period 1, and the probability that these workers observe $i$'s wage, conditional on not receiving information on the firm's type from another worker. The first effect I refer to as `information free-riding''. Workers are willing to pay a cost in the first period to screen, since the information will be useful in the second period. If they are likely to learn the firm's type from other workers then they will be less willing to do so. The second effect is the firms ``reputation effect''. The high type firm will demand a greater premium for being revealed to worker $i$ in the first period if other workers are likely to observe $i$'s wage and learn the firms type. The magnitude of this effect depends on the information sharing parameters, as well as the number of workers that are not screening in the first period. 

If worker $i$ does not screen in the first period, and receives no information about the firm's type from other workers' wages, then $i$ will screen in the second period if and only if
\begin{equation*}
    p w_h + (1-p) \delta^2w(s') \geq w(s').
\end{equation*}
Substituting for $w_h$, this condition becomes
\begin{equation*}
    p w(s'') \geq w(s').
\end{equation*}

\section{Equilibrium in the screening game}

I will focus on equilibria in which the worker never offers a wage below $w(s')$. As Lemma \ref{lem:P1_separating} and the subsequent discussion show, any other equilbrium can be ruled out by assuming PSE, or simply ruling out extreme belief revisions that would support $w^i_2 < w(s')$. It is also necessary to identify how workers who pool will react if they see wage $w^i_1$ from worker $i$ and $\delta w^j_2$ from worker $j$. I assume that in this case the worker infers that the firm is type $s''$.

Let $V_2$ be the value that workers would obtain in the second period if they did not receive any additional information about the firm's type following the first period (so either $V_2 = V_S(p)$ or $V_2 = w(s')$). Let $U_2(s)$ be the corresponding second period value the the type $s$ firm derives from a single worker who has not received information the first period.

Fix the first period separating wage offers $\omega_1 = \{w^i_1\}_{i \in C}$. Let $\pi(A|C,\omega)$ be the payoff of a type $s''$ firm that accepts $w^i_1$ iff $i \in A$, and for $j \in C\setminus A$ offers $\delta w(s')$ (and accepts $w(s')$ for all $i \in W\setminus C$). Then the firms on path payoff is $\pi(C|C,\omega)$, given by
\begin{align*}
   \pi(C|C,\omega) &= \sum_{i \in C} \left( s''-  w^i_1 + \beta (s'' -w(s'')) \right) \\
   &+ \sum_{j \in W\setminus C} \left(s'' - w(s') + \beta\left(P^j(C)(s'' -w(s'')) + (1-P^j(C))U_2(s'') \right)\right).
\end{align*}
For any $A \subseteq C$ we have 
\begin{align*}
    &\pi(A|C,\omega) = \sum_{i \in A} \left( s''-  w^i_1 + \beta (s'' -w(s'')) \right) \\
    & + \sum_{j \in W\setminus C} \left(s''-w(s') + \right. \\
    & \left.\hspace{8mm} \beta\left[P^j(A)(s'' - w(s'')) + (1-P^j(A))\left( P^j(C\setminus A)(s'' - w(s')) +  (1-P^j(C\setminus A)) U_2(s'') \right) \right]\right) \\
    &+ \sum_{j \in C\setminus A}\left( \delta(s'' - \delta w(s')) + \beta\left[P^j(A)(s''-w(s'')) + (1-P^j(A))(s''-w(s'))  \right] \right)
\end{align*}

\begin{lemma}\label{lem:supermodular}
For any $C$ and $\omega$, $\pi(\cdot|C,\omega)$ is supermodular, and strictly supermodular if $\rho_{i,j} \in (0,1)$ for all $i,j$.\footnote{The condition $\rho_{i,j} \in (0,1)$ for all $i,j$ is clearly not necessary for strict supermodularity to hold for a given $C$. Lemma \ref{lem:submodular}) gives the relevant conditions for strict supermodularity.}
\end{lemma}

\begin{proof}
For simplicity I will write $\pi(X)$ rather than $\pi(X|C,\omega)$. I wish to show that $\pi(A\cup B) - \pi(B) \geq  \pi(A) - \pi(A\cap B)$. In all cases, the payoff the firm derives from workers in $A\cap B$ is unchanged. I will consider separately the payoffs of workers in $W\setminus C$, $C\setminus (A \cup B)$, $A \setminus B$ and $B\setminus A$.

Consider the payoff derived from workers in $W \setminus C$. Let $\Pi^j(X) = \prod_{k\in X}(1-\rho_{jk})$. Then the, using the fact that $\Pi^j(X)\Pi^j(C\setminus X) = \Pi^j(C)$, the $W\setminus C$ component of $\pi(X)$ can be written as
{\small
\begin{align*}
    \sum_{j\in W\setminus C}& \left( s'' - w(s') + \beta\left[ (1-\Pi^j(X))(s''-w(s'') ) + \Pi^j(X)(1- \Pi^j(C\setminus X))(s''-w(s')) + \Pi^j(C)U_2(s'') \right] \right)\\
    &= \sum_{j\in W\setminus C} \left( s'' - w(s') + \beta\left[s'' - w(s'') + \Pi^j(X)\left(w(s'') - w(s')\right) + \Pi^j(C)\left( U_2(s'') - s'' + w(s') \right) \right] \right)
\end{align*}
}
Then the $W\setminus C$ component of $\pi(X) - \pi(Y)$ is given by
\begin{equation*}
    \beta\sum_{j\in W\setminus C} \left(\Pi^j(X) - \Pi^j(Y) \right)\left( w(s'')-w(s') \right)
\end{equation*}
For all $j$, $\Pi^j(A\cup B) - \Pi^j(B) \geq \Pi^j(A) - \Pi^j(A\cap B)$ (with strict inequality whenever the conditions in Lemma \ref{lem:submodular} for strict inequality are satisfied). This shows supermodularity on the $W\setminus C$ component of firm payoffs.

Now for the $C\setminus (A\cup B)$ component of $\pi(X)$. These are given by
\begin{equation*}
    \sum_{j\in C\setminus(A\cup B)} \left( \delta(s'' - \delta w(s')) + \beta\left[s'' - w(s') + P^j(X)\left(w(s') - w(s'')  \right) \right] \right).
\end{equation*}
Then the $C\setminus (A\cup B)$ component of $\pi(X) - \pi(Y)$ is given by 
\begin{equation*}
    \beta \sum_{j\in C\setminus(A\cup B)} \left(P^j(X) - P^j(Y)\right)\left(w(s') - w(s'')  \right) 
\end{equation*}
Then submodularity of $P^j(\cdot)$ (Lemma \ref{lem:submodular}) implies that supermodularity holds for this part of payoffs (strictly when the conditions given in Lemma \ref{lem:submodular} are satisfied). 

Now consider the $A\setminus B$ component of payoffs. In both $\Pi(A\cup B)$ and $\Pi(A)$, these workers are successfully screening, so they generate the same payoff for the firm. Supermodularity here will follow by showing that the $A\setminus B$ component of $\pi(B)$ is less than that of $\pi(A\cap B)$. In $\pi(B)$, this is given by 
\begin{equation*}
    \sum_{j\in A \setminus B} \left( \delta(s'' - \delta w(s')) + \beta\left[ s'' - w(s') + P^j(B)\left( w(s') - w(s'') \right) \right] \right).
\end{equation*}
The expression in $\pi(A\cap B)$ is the same, except that $P^j(A\cap B)$ replaces $P^j(B)$. Since $P^j(B) \geq P^j(A\cap B)$ (with strict inequality as long as $P^j(A\cap B) < 1$ and there exists $k \in B\setminus A$ with $\rho_{jk} > 0$) we conclude that supermodularity holds on $A\setminus B$.

Finally, consider the $B\setminus A$ component of payoffs. Since these workers successfully screen in both $\pi(A\cup B)$ and $\pi(B)$, we need only show that the $B\setminus A$ component of payoffs is greater in $\pi(A\cap B)$ than in $\pi(A)$. This follows from the same argument given for the $A\setminus B$ component above.  

\end{proof}

I will first characterize the wage that each screening worker offers to the high type firm. The subtlety here is that the binding incentive constraint under which type $s''$ accepts a wage is generally not the ``single deviation'' of rejecting $i$ and accepting all other $k \in C$. To see this, consider an equilibrium in which the set of screening workers is $C$, and these workers make initial offers of $\omega = \{w_1^i \}_{i\in C}$. $C$ and $\omega$ fully characterize on-path play in equilibrium. For each $i\in C$, let 
\begin{equation*}
    \chi^i = \argmax_{X \subseteq C\setminus i}  \pi(X|C,\omega).
\end{equation*}
$\chi^i$ need not be single valued, but when it will not cause confusion I will discuss as if it is, and refer to this set as $X^i$. $X^i$ is the set of screening offers that type $s''$ will accept if it rejects $i$'s initial offer. If $X^i \neq \varnothing$, let $k \in X^i$. By definition of $X^i$, $\pi(X^i|C,\omega) \geq \pi(X^i\cap X^k|C,\omega)$. But then Lemma \ref{lem:supermodular} implies $\pi(X^i \cup X^k|C,\omega) \geq \pi(X^k|C,\omega)$. 

Suppose $X^i$ and $X^k$ satisfy the conditions for strict supermodularity of $\pi$. Then we have $\pi(X^i \cup X^k|C,\omega) > \pi(X^k|C,\omega)$. Consider what happens if $k$ offers a slightly higher wage. This can only make the payoff from rejecting $k$'s offer worse, since by rejecting $w_1^k$ the firm would have signaled that it was type $s'$ (it could be strictly worse, since $k$ may reject the firms proposal of $\delta w(s')$ in the next round). If the type $s''$ firm accepts $k$'s offer then at worst worker $k$, and any other worker who observes $k$'s wage, will believe that the firm is type $s''$ for sure. But this is the same belief they would hold if the firm accepted $w^k_1$. Thus by raising it's wage offer slightly, worker $k$ makes only a small change to the firm's payoff from accepting the offers of workers in $X^i \cup X^k$. Since $\pi(X^i \cup X^k|C,\omega) > \pi(X^k|C,\omega)$, the firm will still accept $k$'s offer. But then $\omega$ and $C$ cannot characterize an equilibrium, since $k$ has a profitable deviation. The conclusion is that $X^i$ and $X^k$ cannot satisfy the conditions for strict supermodularity. 

\begin{lemma}\label{lem:deviation}
Assume $\rho_{ij} \in (0,1)$ for all $i,j$. In any equilibrium characterized by $C, \omega$, it must be that $\pi(C|C,\omega) = \pi(\varnothing|C,\omega)$, and $\varnothing \in \chi^i$ for all $i\in C$.
\end{lemma}
\begin{proof}
I will show that $\varnothing \in \chi^i$ for all $i \in C$; the remainder of the result then follows from optimality of the worker's wage offers. If $\rho_{ij} \in (0,1)$ for all $i,j$ then the conditions for strict supermodularity of $\pi(\cdot|C,\omega)$ are satisfied for all non-empty $A,B$. This, and the discussion preceding Lemma \ref{lem:deviation}, imply that there is no pair of workers $i,k$ with $k \in X^i$ and $X^k \neq \varnothing$.  

Even without conditions on $\rho$, for any equilibrium $C,\omega$ we will have $\pi(X^i|C,\omega) = \pi(X^k|C,\omega)$ for all $k\in X^i$. To see this, first note that optimality of $w^i_1$ for worker $i$ implies $\pi(C|C,\omega) \leq \pi(X^i|C,\omega)$. Moreover, firm optimality implies $\pi(C|C,\omega) \geq \pi(A|C,\omega)$ for all $A\subseteq C$. Thus $\pi(C|C,\omega) = \pi(X^i|C,\omega) \geq \pi(X^k|C,\omega)$. If $\pi(X^i|C,\omega) > \pi(X^k|C,\omega)$ then firm $k$ can make a slightly higher wage offer that will still be accepted, so we must have $\pi(X^i|C,\omega) = \pi(X^k|C,\omega)$.

Suppose $X^i \neq \varnothing$ and $X^k = \varnothing$ for all $k \in X^i$. Then by the previous claim, $\pi(X^i|C,\omega) = \pi(\varnothing|C,\omega)$, so $\varnothing \in \chi^i$.
\end{proof}

The condition $\pi(C|C,\omega) = \pi(\varnothing|C,\omega)$ pins down $\sum_{i\in C} w^i_1$. This condition is given by
{\small
\begin{align*}
    &\sum_{i \in C} \left( s''-  w^i_1 + \beta (s'' -w(s'')) \right) + \sum_{j \in W\setminus C} \left(s'' - w(s') + \beta\left(P^j(C)(s'' -w(s'')) + (1-P^j(C))U_2(s'') \right)\right) \\
   &= \sum_{i \in C} \left( \delta(s''- \delta w(s')) + \beta (s'' -w(s')) \right) + \sum_{j \in W\setminus C} \left(s'' - w(s') + \beta\left(P^j(C)(s'' -w(s')) + (1-P^j(C))U_2(s'') \right)\right)
\end{align*}
}
This simplifies to 
\begin{equation*}
  |C| \left( \delta(s'' - \delta w(s')) - \left(s'' - \frac{1}{|C|}\sum_{i\in C} w^i_1\right) \right) + \beta \left(|C| + P^j(C)\right)\left( w(s'') - w(s') \right) = 0
\end{equation*}
so
\begin{align*}
    \sum_{i\in C} w^i_1 = \bar{W}(C) &:= |C|\left( (1-\delta)s'' + \delta^2w(s')) \right) - \beta\left( w(s'') - w(s') \right) \left(\bar{P}(C) +|C| \right)\\
    &= |C|\left( (1-\delta^2)w(s'') + \delta^2w(s')) \right) - \beta\left( w(s'') - w(s') \right) \left(\bar{P}(C) +|C| \right)
\end{align*}
where $\bar{P}(C) =\sum_{j \in W\setminus C} P^j(C)$ is the expected number of pooling workers who will observe a wage from a worker in $C$. As expected, the sum of screening wages is decreasing in $\bar{P}(C)$.

The individual wage offers in $\omega$ are constrained by the conditions that $\varnothing \in X^i$ for all $i$. This condition implies that there cannot be too much dispersion in $\omega$.\footnote{Bounds on wages satisfying this condition can be determined by looking recursively at smaller screening sets.} As we will see, it is satisfied automatically, conditional on the sum being $\bar{W}(C)$, in the symmetric equilibria of interest.  

I now turn to the incentives of the workers to engage in screening. If worker $i$ decides not to screen, they must offer a wage that both firm types will accept. By Lemma \ref{lem:P1_upperbound}, the highest such wage is $w(s')$. I will show now that the high type firm will accept an offer of $w(x')$ from any worker in $C$. This is easiest to show under the assumption that workers' prior beliefs are intermediate.

\vst
\noindent \textbf{Intermediate Beliefs.} I say that workers have intermediate beliefs if they screen in the NIP1 game, but not in the NIP2 game. This holds iff
\begin{equation*}
    \dfrac{s'}{s''} \geq p \geq \dfrac{(1-\delta^2)}{(2 - \delta^2 - \beta)\frac{s''}{s'} + \beta-1}
\end{equation*}

\vst
When workers have intermediate beliefs the high type firm is indifferent in the second period between pooling and passing as the low type. This simplifies the analysis. If $p > s'/s''$ the worker screens in the NIP2 game. In this case the high type firm strictly prefers passing as the low type to pooling.

\begin{lemma}\label{lem:screen_deviate}
Assume Intermediate Beliefs. In any equilibrium in which the set of screening workers is $C$, the type $s''$ firm will accept an initial offer of $w(s')$ from any worker in $C$. 
\end{lemma}
\begin{proof}
Let $i$ be the deviating worker. Assume that if the firm rejects the offer of $w(s')$ and proposes $\delta w(s')$ then $i$ believes that the firm is type $s'$. I will show that if this is the case the type $s''$ firm prefers to accept $w(s'')$, so these should indeed be the beliefs under PSE. Similarly if the firm accepts $w(s')$, all workers who observe this will retain their prior belief about the firm's type. 

By Lemma \ref{lem:deviation}, if the firm rejects $i$'s initial offer then it should also reject all offers in $C\setminus i$. Therefore a sufficient condition for the type $s''$ firm to accept $w(s')$ from $i$ is that it prefers accepting $w(s')$ from $i$ to rejecting, given that it is rejecting all other offers in $C$. Given the assumption of intermediate beliefs, from the firm's perspective, convincing workers that it is type $s'$ is the same as getting them to maintain their prior belief. Thus there is no informational difference between accepting $w(s')$ from $i$ or rejecting and proposing $\delta w(s')$. The type $s''$ firm then prefers to accept $w(s')$, since $s'' - w(s') > \delta(s'' - \delta w(s'))$.
\end{proof}

Lemma \ref{lem:screen_deviate} does not imply that if $i\in C$ deviates by proposing $w(s')$, the firm will reject all initial offers in $C\setminus i$. Asymmetry in the initial offers of workers in $C$ may help provide incentives for workers to engage in screening. This is because the distribution of these offers affects the outside option of workers. To see this, suppose for illustrative purposes that $w^j_1 = 0$ for all $j\in C\setminus i$ (ignoring the fact that these workers would then prefer to deviate and offer $w(s')$). If $w^j_1$ were to deviate and offer $w(s')$ then the firm would continue to accept the initial offers of $j \in C\setminus i$. On the other hand, if $w_1^i = 0$ and $i$ offers $w(s')$, the firm will in general reject the initial offers of $j \in C\setminus i$. 

This suggest that when information-sharing parameters differ across workers, screening will be best supported by the following pattern of initial wage offers in $C$ (holding fixed $\sum_{i\in C} w^i_1$): workers with high $P^i(C\setminus i)$, i.e. workers who are likely to observe another wage from $C$, should make low initial wage offers $w^i_1$, while workers with low $P^j(C\setminus j)$ should make high initial wage offers. This is because workers with high are those with the highest ability to free-ride off the information generated by workers in $C\setminus i$. When their offer $w^i_1$ is low, which means $\sum_{C\setminus i} w^j_1 $ is relatively high, a deviation to pooling will lead the firm to reject all initial offers in $C\setminus i$, rendering $i$'s informational advantage useless. Thus, somewhat counter-intuitively, $i$'s informational advantage if it deviates means precisely that it should receive a lower wage on path. 

Given this discussion, it is natural to wonder if asymmetric offers can help support screening even when workers are symmetric. It turns out that this is not the case; if screening can be supported with asymmetric offers than it can be supported with symmetric offers. This will follow from the fact that, with symmetric information-sharing parameters, the \textit{average} initial wage offer of screening workers must be increasing in the size of the set of screening workers. The intuition is that as the set of pooling workers shrinks, it becomes less costly for the high type firm to reveal itself to the screening workers.  

\begin{lemma}\label{lem:avg_wage}
Assume $\rho_{ij} = \rho \in (0,1)$ for all $i,j$. Let $C,\omega$ and $C',\omega'$ be equilibrium screening sets and wage offers, with $|C'| > |C|$. Then $\frac{1}{|C'|}\sum_{i\in C'}w'^i_1 > \frac{1}{|C|}\sum_{i\in C}w^i_1$ 
\end{lemma}
\begin{proof}
Recall that 
\begin{equation*}
    \frac{1}{|C|}\sum_{i\in C}w^i_1 = \left( (1-\delta^2)w(s'') + \delta^2w(s')) \right) - \left(\dfrac{\bar{P}(C)}{|C|} +1\right) \beta\left( w(s'') - w(s') \right).
\end{equation*}
The claim follows in $\bar{P}(C)/|C| > \bar{P}(C')/|C'|$. Under symmetry, $\bar{P}(C) = (|W| - |C|) \left(1-(1-\rho)^{|C|} \right) $. Then the claim is that the function $|C| \mapsto \frac{|W|-|C|}{|C|}(1- (1-\rho)^{|C|})$ is decreasing. The derivative of this function is negative iff
\begin{equation*}
    \dfrac{|W|}{|C|(|W| - |C|)} > -\ln(1-\rho)\dfrac{(1-\rho)^{|C|}}{1 - (1-\rho)^{|C|}}.
\end{equation*}
The right hand side of this expression is bounded above by $1/|C|$ (the limit as $\rho \rightarrow 0$). The inequality follows. 
\end{proof}

As a result of Lemma \ref{lem:avg_wage}, we can see that creating dispersion in the wages of screening workers cannot help their incentives to screen. To see this, suppose that there is a uniform screening wage $w = w^i_1$ for all $i \in C$. If a worker deviates and offers $w(s')$, Lemma \ref{lem:avg_wage} implies that the firm will be unwilling to accept any wages in $C$. This means that the deviating worker anticipates receiving no information. If workers want to deviate in this situation, then adding dispersion only increases the incentives to deviate of the workers who are to make lower initial offers. 

\begin{lemma}\label{lem:no_asymmetry}
Assume Intermediate Beliefs and $\rho_{ij} = \rho \in (0,1)$ for all $i,j$. If screening cannot be supported with a uniform screening wage then in cannot be supported with non-uniform screening wages. 
\end{lemma}

Similarly, if the incentive constraint of the screening workers is violated, i.e. they prefer to deviate and offer $w(s')$, when there is a large set of screening firms, then they will also prefer to deviate when there is a smaller set of screening firms. 

\begin{lemma}\label{lem:single_crossing}
Assume Intermediate Beliefs and $\rho_{ij} = \rho \in (0,1)$ for all $i,j$. If the incentive constraint of screening workers is violated when $N'$ workers screen, then it will be violated if $N<N'$ workers screen. 
\end{lemma}
\begin{proof}
Since screening workers that deviate to pooling anticipate receiving no information from other workers, their payoff from deviation is the same regardless of $|C|$. The claim follows from Lemma \ref{lem:avg_wage}.
\end{proof}

I now turn to the incentives of the low type firm. In equilibrium, the type $s'$ firm receives a payoff of $(1+\beta)\left(s' - w(s') \right)|W|$. This is because the low type firm is indifferent between screening, pooling, and being revealed. In an equilibrium with screening set $C$, the firm could deviate by accepting the initial offers from workers in $A \subseteq C$. This is the so-called ``take the money and run'' deviation. The payoff from doing so, under the assumption that workers who receive mixed signals assume the firm is the high type, is given by\footnote{If workers become convinced that the type $s'$ firm is actually type $s''$, the firm gets a payoff of zero in the second period.}
\begin{equation*}
    (|W| - |A|)(s' - w(s')) + \sum_{i\in A} (s' - w^i_1) + \beta \sum_{i \in W\setminus A} \left( 1 - P^i(A) \right)\left( s' - w(s') \right).
\end{equation*}
The firm prefers not to deviate in this way iff the equilibrium payoff is higher, the condition for which is given by
\begin{equation}\label{eq:low_deviate}
\beta\left( \bar{P}(A) + |A| \right)\left(s' - w(s')\right) \geq \sum_{i\in A} \left( w(s') - w^i_1 \right)
\end{equation}
Here $\bar{P}(A) + |A|$ is the expected number of workers with whom the type $s'$ firm will be unable to reach any agreement in period 2, and $\sum_{i\in A} \left( w(s') - w^i_1 \right)$ is the potential gain in the first period of accepting wages that may be below $w(s')$. Note that the condition in (\ref{eq:low_deviate}) does not depend on $C$ directly, only on $A$ and $\omega$. I say that the low type incentive constraint is satisfied for $C,\omega$ if (\ref{eq:low_deviate}) is satisfied for all $A \subseteq C$.

\begin{lemma}\label{lem:lowfirm_IC}
Let $C$ be the set of screening workers, and suppose $w^i_1 = \bar{W}(C)/|C|$ for all $i\in C$. If $\rho_{i,j} = \rho \in (0,1)$ for all $i,j$ then the low type incentive constraint is satisfied iff (\ref{eq:low_deviate}) holds for $A=C$.
\end{lemma}
\begin{proof}
Rewriting (\ref{eq:low_deviate}), we have 
\begin{equation*}
    \beta \left( \frac{\bar{P}(|A|)}{|A|} + 1\right)\left( s' - w(s')\right) \geq w(s') - \frac{\bar{W}(C)}{|C|} 
\end{equation*}
As shown in the proof of Lemma \ref{lem:avg_wage}, $\bar{P}(A)/|A|$ is decreasing under symmetric information sharing. The claim follows. 
\end{proof}

I now want to consider the incentives of pooling workers. A pooling worker may deviate by trying to screen the firm. Assume that the deviating pooling worker first proposes a wage $\hat{w}$, which it expects to be accepted if and only if the firm is type $s''$. I will show when these beliefs are consistent with the incentives of the firm. 

Suppose there is an equilibrium with $|W|>|C|>0$. If a pooling worker $i \in W\setminus C$ deviates and tries to screen, the firm can either accept $i$'s initial offer or reject and propose $\delta w(s')$ in the second round. 

\begin{lemma}\label{lem:pool_deviate}
Assume $\rho_{i,j} \in (0,1)$ for all $i,j$. In an equilibrium with screening set $C$, if a pooling worker $i$ makes an initial offer of $\hat{w}$, the type $s''$ firm will accept if and only if $\bar{W}(C\cup i) - \bar{W}(C) \geq \hat{w}$. 
\end{lemma}
\begin{proof}
Suppose that the type $s''$ firm rejects $i$'s offer. Then $i$ will believe that the firm is type $s'$ as long as it does not observe evidence to the contrary. Since the second period payoff to the firm from a worker who believes it is type $s'$ is weakly higher when workers maintain their prior, the incentive of the firm to screen with the workers in $C$ is lower than in equilibrium. By Lemma \ref{lem:deviation}, on path the firm weakly prefers rejecting all initial offers in $A$ to accepting all initial offers in $A$, for all $A \subseteq C$. Thus following it's rejection of $i$'s initial offer, the firm will also weakly prefer rejecting all initial offers in $A$ to accepting all initial offers in $A$, for all $A \subseteq C$. (Under Intermediate Beliefs the firm is indifferent between rejecting all initial offers in $C$ and accepting all initial offers in $C$, both on path and following the rejection of $i$'s initial offer.)

By the definition of $\bar{W}(C\cup i)$, the firm is indifferent between accepting all initial offers in $C \cup i$ and rejecting all such offers. Combined with the preceding paragraph, this implies that if $i$ makes an initial offer of $\bar{W}(C\cup i) - \bar{W}(C)$, the type $s''$ firm will accept. In particular, the firm will prefer accepting all initial offers in $C\cup i$ to rejecting all initial offers in $C\cup i$. Moreover, given acceptance of $i$'s offer, the high-type firm is more willing to be screened by other workers than it was with no deviation by $i$. 

It remains to show that no higher offer by worker $i$ will be accepted. This would only be the case if there was some set $A \subset C$ such that the firm strictly preferred accepting only the initial offers in $A\cup i$ to accepting all initial offers in $C \cup i$. This is ruled out however by submodularity of $P^j$ (Lemma \ref{lem:submodular}), using a similar argument to that given for Lemma \ref{lem:supermodular}. (It does not follow directly from Lemma \ref{lem:supermodular} because we are comparing a case in which $i$ pools to one in which it screens. However the difference in the argument is minor.)
\end{proof}

Lemma \ref{lem:pool_deviate} does not imply that the deviating pooling worker will be able to successfully screen by offering $\bar{W}(C\cup i) - \bar{W}(C)$: it may be that the low type firm will also want to accept. Mild conditions on primitives guarantee that this is not the case. The following Lemma is an intermediate step. It narrows down the set of possible responses by the low type firm to the pooling workers $i$'s deviation; either the low type firm rejects $i$'s initial offer, or it accepts all initial offers.  

\begin{lemma}\label{lem:pool_deviate2}
Assume $\rho_{i,j} = \rho \in (0,1)$ for all $i,j$ and the low type incentive constraint is satisfied with screening set $C$ and $w^i_1 = \bar{W}(C)/|C|$ for all $i\in C$. If $i$ deviates by offering $\hat{w} = \bar{W}(C\cup i) - \bar{W}(C)$ and the type $s'$ firm accepts, then it must be that the firm also accepts all initial wage offers in $C$.\footnote{It is easy to see from condition (\ref{eq:low_deviate}) that weaker conditions can be given for the conclusion, but they are not as convenient to state. I present this version since it is what will be used later on.}
\end{lemma}
\begin{proof}
As shown in the proof of Lemma \ref{lem:avg_wage}, with symmetric information sharing $\bar{W}(C)/|C|$ is decreasing in $|C|$. This implies that $\bar{W}(C\cup i) - \bar{W}(C) > \bar{W}(C)$. Suppose that in response to $i$'s deviation, the type $s'$ firm accepts all initial offers in $i \cup S$ for $S \subset C$, (or $S = \varnothing$). This means that condition (\ref{eq:low_deviate}) is violated for $A = i \cup S$. Since $\bar{W}(C\cup i) - \bar{W}(C) > \bar{W}(C)$, this meas there also exists a set $X \subset S$ with $|X| = |S\cup i|$ such that $\sum_{k \in X} w^k_1 < \bar{W}(C\cup i) - \bar{W}(C) + \sum_{k \in S} w^k_1$. But then (\ref{eq:low_deviate}) does not hold for $A =X$, so this cannot be an equilibrium.
\end{proof}

Using Lemma \ref{lem:pool_deviate2}, we can identify conditions under which the low type firm will not accept the initial wage offer of a pooling worker who deviates to offer $\bar{W}(C\cup i) - \bar{W}(C)$. In other words, the pooling worker will be able to screen with this deviation.

\begin{lemma}\label{lem:pool_deviate3}
Assume $\rho_{i,j} = \rho \in (0,1)$ for all $i,j$ and the low type incentive constraint is satisfied with screening set $C$ and $w^i_1 = \bar{W}(C)/|C|$ for all $i\in C$. If $j \in W\setminus C$ deviates by offering $\hat{w} = \bar{W}(C\cup i) - \bar{W}(C)$ then the low type will reject this initial offer iff 
\begin{equation*}
    \beta\left(\dfrac{\bar{P}(C\cup j)}{|C|+1} + 1 \right) \geq (1-\delta)\frac{s' - s''}{s' - w(s'')}.
\end{equation*}
\end{lemma}
\begin{proof}
By construction, the type $s''$ firm is indifferent between accepting all initial offers in $C\cup j$, given $j$'s proposal of $\bar{W}(C\cup j) - \bar{W}(C)$, and rejecting all these offers. This indifference condition is given by 
\begin{equation}\label{eq:pf1}
    \beta\left( \bar{P}(C\cup j) + |C| + 1 \right)\left( w(s'') - w(s') \right) = (|C|+1)\left( s'' - \frac{1}{|C|+1}\bar{W(C\cup j)} - \delta(s'' - \delta w(s') \right).
\end{equation}
The left hand side of this expression is the reputational benefit of delay, and the right hand side is the period 1 cost. 

By Lemma \ref{lem:pool_deviate2}, if the low type firm accepts $j$'s initial offers then it accepts all initial offers in $C \cup j$. On the other hand, the fact that condition (\ref{eq:low_deviate}) was satisfied for all $A\subseteq C$ implies that if the low type rejects $j$'s initial offer than it also rejects all initial offers in $C$. Thus the low type rejects $j$'s initial offer if and only if it prefers rejecting all initial offers in $C\cup j$ to rejecting all such offers. The condition for this is 
\begin{equation}\label{eq:pf2}
    \beta\left( \bar{P}(C\cup j) + |C| + 1 \right)(s' - w(s')) \geq (|C| + 1)\left( s' - \frac{1}{|C|+1}\bar{W}(C\cup j) - \delta(s'-\delta w(s') \right).
\end{equation}
Combining conditions (\ref{eq:pf1}) and (\ref{eq:pf2}) yields the desired condition.
\end{proof}

Since $\bar{P}(C)/|C|$ is decreasing, it is sufficient to check the condition of Lemma \ref{lem:pool_deviate3} for $|C| = |W|-1$, in which case it reduces to 
\begin{equation*}
    \beta \geq (1-\delta)\frac{s' - s''}{s' - w(s'')}. 
\end{equation*}

A worker in $W\setminus C$ can deviate and offer at most $\bar{W}(C\cup i) - \bar{W}(C)$ if they hope the high type firm to accept. Under the conditions of Lemma \ref{lem:pool_deviate3}, this leads to successful screening; the firm will accept this initial offer if and only if it is type $s''$. Therefore the worker's belief that offering $\bar{W}(C\cup i) - \bar{W}(C)$ will lead to screening is consistent, and this constitutes the best deviation from pooling. 

Under the conditions of Lemma \ref{lem:pool_deviate3}, $\bar{W}(C)/|C|$ is decreasing. Therefore a pooling worker who deviates when $|C|=N$ does better than a screening worker in the symmetric-offers equilibrium when $|C| = N+1$. Since the payoff of screening workers is increasing in $|C|$, this means that in equilibrium pooling workers always do better than screeners. 

\subsection{Properties of equilibrium}

The properties of equilibrium can be summarized as follows, beginning with those that hold under the least restrictive assumptions. 

\vst
\noindent Assume $\rho_{i,j} \in (0,1)$ for all $i,j$.
\begin{itemize}
    \item Lemma \ref{lem:deviation}. Interpretation: Fix the set $C$ of workers who will screen in equilibrium, and consider \textit{a)} the choice of initial wage offers by workers in $C$, and \textit{b)} the incentives of the high type firm to accept the initial offer, rather than delay to pass as the low type. Supermodularity in information transmission implies that each worker knows that if they increase their offer the firm will reject not just theirs, but all other initial offers in $C$. Thus the workers know they will get no information about the firm's type. This pins down the sum of the initial offers for workers in $C$. 
\end{itemize}

\vst
\noindent Assume $\rho_{i,j} \in (0,1)$ for all $i,j$ and Intermediate Beliefs. 
\begin{itemize}
    \item Lemma \ref{lem:screen_deviate}. Interpretation: Any screening worker can deviate to pooling by making an initial offer of $w(s')$.
\end{itemize}

\vst
\noindent Assume $\rho_{ij} = \rho \in (0,1)$ for all $i,j$.\footnote{Definitely not necessary for Lemma \ref{lem:avg_wage}}
\begin{itemize}
    \item Lemma \ref{lem:avg_wage}. Interpretation: given the supermodularity of information sharing, there are two positive externalities of screening. First, pooling workers may get some information. Second, other screening workers need to pay lower information rents. 
    \item Screening is fragile; if any screening worker deviates no screening occurs. 
\end{itemize}

\vst
\noindent Assume $\rho_{ij} = \rho \in (0,1)$ for all $i,j$ and symmetry of screening wage offers ($w^i_1 = \bar{W}(C)/|C|$ for all $i \in C$).
\begin{itemize}
    \item In this case Lemma \ref{lem:avg_wage} implies that when all workers in $C$ make the same initial wage offer $\bar{W}(C)/|C|$, we have $X^i = \varnothing$ for all $i\in C$. 
    \item Lemma \ref{lem:lowfirm_IC}. Interpretation: The binding constraint for the low type firm is to not accept all initial offers in $C$. 
    \item Lemma \ref{lem:pool_deviate3}. Interpretation: Under the condition given in the lemma, workers who are supposed to pool in equilibrium can deviate to screening by making an offer of $\bar{W}(C\cup i) - \bar{W}(C)$. By Lemma \ref{lem:pool_deviate}, this is the relevant incentive constraint for workers who pool in equilibrium. 
\end{itemize}

\vst
\noindent Assume $\rho_{ij} = \rho \in (0,1)$ for all $i,j$ and Intermediate Beliefs.
\begin{itemize}
    \item Lemma \ref{lem:single_crossing}. Interpretation: Incentives of workers to screen are stronger when more workers screen. 
\end{itemize}

In the intersection of the restrictions, we have the following (partial) characterization of equilibrium. 

\vst
\noindent Assume $\rho_{ij} = \rho \in (0,1)$ for all $i,j$,  Intermediate Beliefs, and symmetric screening wage offers. The incentive constraints in an equilibrium in which the screening set is $C$ are as follows.
\begin{itemize}
    \item Screening workers: either screen by offering $\bar{W}(C)/|C|$ or deviate to pooling by offering $w(s')$, in which case the high type firm will reject all other initial offers from screening workers, and so there will be no information revealed by wages. 
    \item Pooling workers: pool and get wage $w(s')$ or deviate to screening by making initial wage offer $\bar{W}(C\cup i) - \bar{W}(C)$. Following a deviation screening will continue with all workers in $C\cup i$.
    \item High type firm. Accept all initial offers of workers in $C$ or deviate and reject all initial offers of workers in $C$. The firm will be indifferent as long as all workers offer $\bar{W}(C)/|C|$.
    \item Low type firm. Reject all initial offers in $C$ or deviate and accept all initial offers in $C$.
\end{itemize}
Moreover, equilibrium has the following properties.
\begin{itemize}
    \item Pooling workers are always better off than screening workers. 
    \item The payoffs of both pooling and screening workers are increasing in $|C|$. 
    \item Screening is fragile; if any screening worker deviates then no screening occurs. 
\end{itemize}

\subsection{Key takeaways}

The key constraint is that screening workers must prefer screening to offering $w(s')$ and receiving no information. When information sharing increases this constraint eliminates equilibria with smaller screening sets. As long as the incentive constraints of the pooling worker and the low type firm do not bind, this is the comparative static.

\end{document}